\documentclass[8pt]{article}
\usepackage{geometry}               
\usepackage{epstopdf}
\usepackage{amsthm}
\usepackage{amsfonts}
\usepackage{MnSymbol}
\usepackage{hyperref}
\usepackage{color}
\usepackage{tikz}
\usepackage{graphicx}
\usepackage[affil-it]{authblk}
\usepackage[applemac]{inputenc} 
\usepackage[toc,page]{appendix}

\newtheorem{theorem}{Theorem}[section]
\newtheorem{lemma}[theorem]{Lemma}
\newtheorem{proposition}[theorem]{Proposition}

\newtheorem{definition}[theorem]{Definition}
\newtheorem{remark}[theorem]{Remark}

\numberwithin{equation}{section}

\newcommand{\bea}{\begin{eqnarray}}
\newcommand{\eea}{\end{eqnarray}}
\def\beaa{\begin{eqnarray*}}
\def\eeaa{\end{eqnarray*}}
\def\ba{\begin{array}}
\def\ea{\end{array}}
\def\be#1{\begin{equation} \label{#1}}
\def \eeq{\end{equation}}
\def\bsplit{\begin{split}}

\def\c{\cdot}

\def\dual{{\,^\star \mkern-2mu}}
\def\tr{\mbox{tr}}

\newcommand{\nabb}{\nab\mkern-13mu /\,}

\def\Db{\dot{\D}}
\def\squared{\dot{\square}}

\renewcommand{\div}{\mbox{div }}

\newcommand{\divv}{\mbox{div}\mkern-19mu /\,\,\,\,}
\newcommand{\lapp}{\mbox{$\bigtriangleup  \mkern-13mu / \,$}}

\def\nab{\nabla}

\def\pr{\partial}

\def\DDs{ \, \DD \hspace{-2.4pt}\dual    \mkern-16mu /}

\def\dds{ \, d  \hspace{-2.4pt}\dual    \mkern-14mu /}
\def\ddd{ \,  d \hspace{-2.4pt}    \mkern-6mu /}

\def\rhod{ \,^\star  \hspace{-2.2pt} \rho}

\def\a{\alpha}
\def\b{\beta}
\def\ga{\gamma}

\def\de{\delta}

\def\ep{\epsilon}
\def\la{\lambda}

\def\si{\sigma}

\def\om{\omega}
\def\Om{\Omega}

\def\th{{\theta}}

\def\ka{\kappa}
\def\ze{\zeta}
\def\Up{\Upsilon}

\def\vphi{{\varphi}}
\def\vth{\vartheta}

\def\vthb{\underline{\vth}}
\renewcommand{\aa}{\protect\underline{\a}}
\newcommand{\bb}{\protect\underline{\b}}
\def\omb{{\underline{\om}}}

\newcommand{\chib}{\underline{\chi}}
\newcommand{\xib}{\underline{\xi}}
\newcommand{\etab}{\underline{\eta}}

\def\kab{\underline{\kappa}}

\def\DD{{\mathcal D}}

\def\D{{\bf D}}

\def\F{{\bf F}}
\def\G{{\bf G}}

\def\M{{\bf M}}

\def\R{{\bf R}}

\def\T{{\bf T}}

\def\W{{\bf W}}

\def\Z{{\bf Z}}

\def\g{{\bf g}}

\def\fb{\underline{f}}

\def\qf{\frak{q}}

\def\ff{\frak{f}}

\def\chih{\widehat{\chi}}
\def\chibh{\widehat{\chib}}
\def\trch{\tr \chi}
\def\trchb{\tr \chib}

\def\aS{\,^{(1+3)} \hspace{-2.2pt}\a} 
\def\bS{\,^{(1+3)} \hspace{-2.2pt}\b} 
\def\rhoS{\,^{(1+3)} \hspace{-2.2pt}\rho} 
\def\rhodS{\,^{(1+3)} \hspace{-2.2pt}\rhod} 
\def\bbS{\,^{(1+3)} \hspace{-2.2pt}\bb} 
\def\aaS{\,^{(1+3)} \hspace{-2.2pt}\aa} 

\def\chiS{\,^{(1+3)} \hspace{-2.2pt}\chi} 
\def\chibS{\,^{(1+3)} \hspace{-2.2pt}\chib} 
\def\chihS{\,^{(1+3)} \hspace{-2.2pt}\chih} 
\def\chibhS{\,^{(1+3)} \hspace{-2.2pt}\chibh}
\def\trchS{\,^{(1+3)} \hspace{-2.2pt}\trch} 
\def\trchbS{\,^{(1+3)} \hspace{-2.2pt}\trchb}
\def\xiS{\,^{(1+3)} \hspace{-2.2pt}\xi} 
\def\xibS{\,^{(1+3)} \hspace{-2.2pt}\xib} 
\def\etaS{\,^{(1+3)} \hspace{-2.2pt}\eta} 
\def\etabS{\,^{(1+3)} \hspace{-2.2pt}\etab} 
\def\zeS{\,^{(1+3)} \hspace{-2.2pt}\ze} 
\def\omS{\,^{(1+3)} \hspace{-2.2pt}\om} 
\def\ombS{\,^{(1+3)} \hspace{-2.2pt}\omb}

\def\bF{\,^{(F)} \hspace{-2.2pt}\b}
\def\bbF{\,^{(F)} \hspace{-2.2pt}\bb}
\def\aF{\,^{(F)} \hspace{-2.2pt}\a}
\def\aaF{\,^{(F)} \hspace{-2.2pt}\aa}
\def\rhoF{\,^{(F)} \hspace{-2.2pt}\rho}
\def\rhodF{\,^{(F)} \hspace{-2.2pt}\rhod}

\def\c{\cdot}
\def \f12{\frac 1 2 }
\def\ov{\overline}
\def\Db{\dot{\D}}
\def\squared{\dot{\square}}

\def\err{\mbox{Err}}

\def\gS{ g \mkern-8.5mu/\,}

\def\err{\mbox{Err}}

\begin{document}
\title{\large{\textbf{Coupled gravitational and electromagnetic perturbations of Reissner-Nordstr{\"o}m spacetime in a polarized setting}}}
\author{\normalsize{Elena Giorgi}}
\date{} 
\maketitle
\begin{abstract}
We derive a system of equations governing the coupled gravitational and electromagnetic perturbations of Reissner-Nordstr{\"o}m spacetime. The equations are derived in the context of global non-linear stability of Reissner-Nordstr{\"o}m under axially symmetric polarized perturbations, as a generalization of the recent work on non-linear stability of Schwarzschild spacetime of Klainerman-Szeftel (\cite{stabilitySchwarzschild}). The main result consists in deriving,  through a Chandrasekhar-type transformation, a gauge invariant quantity associated to the electromagnetic tensor that verifies a Regge-Wheeler equation. In this paper, we present the derivation of the main equations, postponing the estimates of the non-linear system and the construction of the dynamical spacetime to a later time.
\end{abstract}

\tableofcontents

\section*{Introduction}

\addcontentsline{toc}{section}{Introduction}

 One of the fundamental open problems in General Relativity is the one concerning the stability of the exterior of black holes under gravitational perturbations. The stability conjecture says that the class of metrics of the Kerr family is stable under small perturbations of initial data as solutions to the vacuum Einstein equation: 
 \bea\label{vacuum}
\operatorname{Ric}(g)=0,
\eea 
where $\operatorname{Ric}(g)$ is the Ricci curvature tensor of the metric $g$. The conjecture is also formulated for charged black holes, and in that case it conjectures that the Kerr-Newman family of spacetimes is stable under small perturbation of initial data as solutions to the Einstein-Maxwell equation:
\bea \label{Einsteineq}
\operatorname{Ric}(g)_{\mu\nu}=T(F)_{\mu\nu}:=2 F_{\mu \lambda} {F^{\lambda}}_\nu - \frac 1 2 g_{\mu\nu} F^{\a\b} F_{\a\b}
\eea
where $F$ is a $2$-form satisfying Maxwell's equations 
\bea \label{Max}
\nabla_{[\a} F_{\b\gamma]}=0, \qquad \nabla^\a F_{\a\b}=0.
\eea
In the case of Einstein-Maxwell equation, coupled gravitational and electromagnetic perturbations of initial data are considered, i.e. the Weyl curvature and the electromagnetic tensor $F$ of the spacetime are both perturbed with respect to the initial spacetime.

The only known result concerning the full non-linear stability of a vacuum spacetime without symmetries is the celebrated stability of Minkowski space by Christodoulou-Klainerman (\cite{Ch-Kl}). The result was generalized by Zipser to the non-linear stability of Minkowski space under gravitational and electromagnetic perturbations in \cite{Zipser}. The first linear stability of a black hole spacetime has been proved by Dafermos-Holzegel-Rodnianski in \cite{DHR}, in their proof of linear stability of Schwarzschild spacetime. The full non-linear stability of the Schwarzschild solution is still open, and would require a formulation of the stability not just for Schwarzschild, but for slowly rotating Kerr solutions. Indeed, even if one restricts to small perturbations of Schwarzschild, one should expect that generically the spacetime would evolve to a slowly rotating Kerr spacetime, with small but non-zero angular momentum.

A recent work by Klainerman-Szeftel (\cite{stabilitySchwarzschild}) addresses the global non-linear stability of Schwarzschild black hole as solution to the Einstein vacuum equation \eqref{vacuum}. The authors in \cite{stabilitySchwarzschild} consider a particular class of gravitational perturbations of Schwarzschild, namely axially symmetric polarized spacetimes. This choice of perturbations forces the final state of the evolution to have zero angular momentum, therefore ruling out the general Kerr spacetime. They can therefore prove the global non-linear stability of Schwarzschild space, meaning that the axially symmetric polarized perturbed spacetime close to the initial Schwarzschild evolves to another Schwarzschild solution, with mass close to initial one.

In \cite{stabilitySchwarzschild}, the authors consider the celebrated Teukolsky equation (first derived in \cite{Teukolsky}) verified by the extreme component of the Riemann curvature $\alpha$, and apply a Chandrasekhar-type transformation to obtain a new quantity $\qf$, at the level of second derivative of $\a$, that verifies a Regge-Wheeler equation. This transformation was first used in the context of linear stability of Schwarzschild in \cite{DHR} to derive decay estimates for solution of Teukolsky equation. In \cite{stabilitySchwarzschild}, the decay estimates for $\qf$ are the starting point to derive decays for the curvature components and the connection coefficients, using the null structure equations and the Bianchi identities. The dynamical construction of the spacetime follows, along with many subtleties related to the non-linearity of the problem.

In the present paper, we address the problem of stability of charged black holes subject to coupled gravitational and electromagnetic perturbations. A particular class of these spacetimes are the spherically symmetric charged black holes, namely the Reissner-Nordstr{\"o}m solution, corresponding to Kerr-Newman spacetime with zero angular momentum. In local coordinates $(t,r,\th, \vphi)$ the metric is expressed as
\beaa
\g_{RN}= -\left(1-\frac{2m}{r}+\frac{Q^2}{r^2}\right) dt^2+\left(1-\frac{2m}{r}+\frac{Q^2}{r^2}\right)^{-1}dr^2 +r^2(d\th^2+\sin^2\th d\vphi^2)
\eeaa
where $m$ is the mass and $Q$ is the charge of the black hole. Gravitational and electromagnetic perturbations of Reissner-Nordstr{\"o}m black holes have been extensively considered in the setting of metric perturbations. Moncrief (\cite{Moncrief}) reduced the governing equations to pair of decoupled one dimensional wave equations both for odd and for the even parity perturbations. This approach corresponds to the description of gravitational perturbations of Schwarzschild spacetime via Regge-Wheeler and Zerilli equations for metric perturbations (as in \cite{MuTao}), as opposed to the Teukolsky equations for curvature perturbations (as in \cite{DHR}). Chandrasekhar (\cite{Chandra2}), using the Newman-Penrose formalism, derived a pair of decoupled equations for perturbations of Reissner-Nordstr{\"o}m which can be transformed to one dimensional wave equations for both parity perturbations. In order to decouple the equations, Chandrasekhar chooses a gauge (the phantom gauge, as discussed in Remark \ref{remChandra}) and separates them in radial and angular parts. Our approach instead is based on the use of a gauge invariant quantity, and no separation of variables is needed to decouple the equations. 

As for the vacuum Einstein equation, the general problem of stability of charged black holes would have to be solved in the context of global non-linear stability of the Kerr-Newman spacetime as solution to the Einstein-Maxwell equation \eqref{Einsteineq}. However, as in the case of axially symmetric polarized perturbations of Schwarzschild in \cite{stabilitySchwarzschild}, restricting to axially symmetric polarized perturbations of Reissner-Nordstr{\"o}m will force the final state of the evolution to a non-rotating black hole, excluding the general Kerr-Newman spacetime.

The present work is meant to be the initial step in order to extend the stability result of \cite{stabilitySchwarzschild} to the case of electrovacuum perturbations of charged black holes. The equations governing the evolution of the curvature and the electromagnetic tensor are now coupled. In particular, the Teukolsky equation verified by the extreme component of the Weyl curvature $\a$ is coupled with the electromagnetic components coming from the non-vanishing Ricci curvature. Applying a Chandrasekhar-type transformation we derive the corresponding new quantity $\qf$ veryfing a Regge-Wheeler type equation, coupled with electromagnetic terms. However, these additional terms are multiplied by the charge of the spacetime, so, provided that we have control on the electromagnetic part, they could in principle be absorbed as error terms for small enough charge. 

The main new insight is the control of the electromagnetic part making use of a gauge invariant quantity depending on the electromagnetic components, whose null derivative appears in the Teukolsky equation for $\a$. Applying a new Chandrasekhar-type transformation, at the level of one derivative only as opposed to two derivatives as in the case of curvature, we are able to find a new quantity $\qf^{\F}$ verifying another Regge-Wheeler type equation, coupled with the curvature term $\qf$. Therefore, we have a coupled system of wave equations for the term encoding the curvature $\qf$ and for the term encoding the electromagnetic part $\qf^\F$ which at the linear level looks like\footnote{The explicit form of the final system is \eqref{finalsystem}}
\bea \label{schemesystem}
\begin{cases}
\square_\g \qf= V_1 \qf+ e \c \mathcal{M}(\partial^{\leq 2} \qf^{\F})+ e(l.o.t.(\qf^\F))+e^2(l.o.t.(\qf))\\
\square_\g \qf^{\F}=V_2 \qf^{\F}+e \c \mathcal{M}(\qf) +e^2(l.o.t.(\qf^\F))
\end{cases}
\eea
where the operator $\square_\g=\D^\a \D_\a$ is the wave operator associated to the perturbed metric $\g$, $e$ is the charge of the perturbed spacetime, and $\mathcal{M}$ is an expression of the arguments. We denote $l.o.t.(\qf)$ and $l.o.t.(\qf^\F)$ lower order terms with respect to $\qf$ and $\qf^\F$ respectively.

The two wave equations are coupled: on the right hand side of the equation for the curvature term $\qf$ we find an expression of the electromagnetic term $\qf^\F$, and similarly on the right hand side of the equation for $\qf^\F$ we find the curvature term $\qf$. Notice that those coupled terms are multiplied by the charge of the spacetime. The coupling is not symmetric in terms of dependence on derivatives: the presence of two derivatives of $\qf^{\F}$ on the right hand side of the first equation is a consequence of the Teukolsky equation for $\a$, in which the derivative of the electromagnetic term appears. However, this asymmetry is good in terms of deriving estimates for such a system. Indeed, taking one derivative of the second equation and deriving Morawetz estimates for it, we would have a term for second derivative of $\qf$ and one term for first derivative of $\qf$, the latter multiplied by the charge. Those are exactly the kind of terms appearing in the Morawetz estimates obtained for the first equation (because of the presence of $\partial^{\leq 2}\qf^\F$), with the term for second derivative of $\qf^\F$ multiplied by the charge. Summing those estimates, the terms on the right hand side multiplied by the charge could be absorbed on the right hand side for small enough charge. In addition to the coupling, there is the presence of lower order terms, which will have to be treated either in the spirit of \cite{Siyuan} in the case of slowly rotating Kerr (i.e. considering a system of equations for the lower order terms), or as in \cite{TeukolskyDHR} (i.e. deriving decay for the lower order terms using transport equations). The analysis of the system, complete of Morawetz and $r^p$-estimates will follow in a later paper.

In this paper, we derive the main equations leading to the system \eqref{schemesystem}, as a first step towards the proof of non-linear stability of Reissner-Nordstr{\"o}m spacetime under polarized perturbations. We remark that the structure of the system \eqref{schemesystem} doesn't depend on the polarization of the metric nor on the assumption of axial symmetry, as proved by computations done by computer algebra by Steffen Aksteiner. We present the general result in the Appendix. The system can therefore be used to prove linear stability of Reissner-Nordstr{\"o}m spacetime.

\textbf{Acknowledgements} The author is grateful to Sergiu Klainerman for proposing the extension of the work in \cite{stabilitySchwarzschild} to electrovacuum spacetimes, and for his continuous encouragement. The author thanks J{\'e}r{\'e}mie Szeftel and Steffen Aksteiner for useful discussions, and Steffen Aksteiner for his remarks on the importance of signature zero quantities and operators, and for performing the general computations, sketched in the Appendix.

\section{Electrovacuum axially symmetric polarized spacetimes}
We consider electrovacuum spacetimes $(\M, \g)$, namely solution to the Einstein-Maxwell equation
\bea 
\operatorname{Ric}(\g)_{\mu\nu}=T(\F)_{\mu\nu}:=2 \F_{\mu \lambda} {\F^{\lambda}}_\nu - \frac 1 2 g_{\mu\nu} \F^{\a\b} \F_{\a\b}
\eea
where $\F$ is a $2$-form satisfying Maxwell's equations 
\bea
\D_{[\a} \F_{\b\gamma]}=0, \quad \D^\a \F_{\a\b}=0
\eea
We denote $\D$ the Levi-Civita connection of the spacetime $(\M, \g)$.

 An axially symmetric spacetime $(\M,\g, \Z) $  is a four dimensional   simply connected manifold $\M$   with a Lorentzian metric $\g$  and an  axial  Killing vectorfield $\Z$  on $\M$ that preserves $\F$, i.e. $\mathcal{L}_{\Z} \F=0$.

The Ernst potential of the spacetime    is given by
 \beaa
 \si_\mu&:=&\D_\mu(-\Z^\a\Z_\a)- i\in_{\mu\b\ga\de}\Z^\b \D^\ga\Z^\de.
 \eeaa
 The $1$-form   $\si_\mu dx^\mu$ is closed   and thus  there exists a function $\si:\M \to\mathbb{C}$, called the $\Z$- Ernst potential,
  such that  $ \si_\mu =\D_\mu\si.$
   Note also  that    $\D_\mu\g(\Z, \Z)=2 \G_{\mu\la} \Z^\la = - Re(\si_\mu) $   where $\G_{\mu\nu}=\D_{\mu}\Z_{\nu}$.   Hence 
 we can choose  the potential $\si$  such that $ Re(\si)=- X$, where $X=\g(\Z, \Z)$.
  \begin{definition}
 An axially symmetric   Lorentzian manifold $(\M, \g, \Z)$       is said to be polarized  if   the Ernst potential  $\si$ is real, i.e.       $\si=-X$.  
In this case, we can find coordinates $(\vphi, x^a)$ such that $\Z=\partial_{\vphi}$, and the metric $\g$ can be written in the form
\bea \label{formmetric}
\g= Xd\vphi^2+g_{ab} dx^a dx^b=e^{2\Phi} d\vphi^2 +g_{ab} dx^a dx^b
\eea
where $g_{ab}$ is a $1+2$ Lorentzian metric, and $\Phi=\frac 1 2 \log (X)$. The axial symmetry implies that $X$ and $g$ are independent of $\vphi$. 
 \end{definition}

Following \cite{Ch-Kl}, we define the null decomposition of the curvature and the electromagnetic tensor on a given $\Z$-invariant polarized $S$-foliation of an electrovacuum spacetime. We assume we have a fixed adapted null   pair   $e_3, e_4$, i.e.   future directed  $\Z$-invariant    null vectors orthogonal to  the leaves $S$  of the foliation,   such as $\g(e_3, e_4)=-2$, while on $S$ we have an orthonormal frame $e_1, e_2$. 
 
 We   define the spacetime Ricci coefficients, where the indices $A,B$ take values $1,2$
   \bea\label{def1}
   \begin{split}
   \chiS_{AB}:&=\g(\D_A  e_4, e_B), \qquad \xiS_A:=\frac 1 2 \g(\D_4 e_4, e_A), \qquad \etaS_A:=\frac 1 2 \g(\D_3 e_4, e_A)\\
   \zeS_A:&=\frac 1 2 \g( \D_A e_4, e_3), \qquad \omS:=\frac 1 4 \g(\D_4 e_4, e_3)
   \end{split}
   \eea
   and interchanging $e_3, e_4$,
    \bea\label{def2}
   \begin{split}
   \chibS_{AB}:&=\g(\D_A  e_3, e_B), \qquad \quad \xibS_A:=\frac 1 2 \g(\D_3 e_3, e_A), \qquad \etabS_A:=\frac 1 2 \g(\D_4 e_3, e_A)\\
   \zeS_A:&=-\frac 1 2 \g( \D_A e_3, e_4), \qquad \ombS:=\frac 1 4 \g(\D_3 e_3, e_4)
   \end{split}
   \eea
We define the spacetime null curvature components of the Weyl curvature $\W$,  
   \bea\label{def3}
\begin{split}
\aS_{AB}:&=\W_{A4B4}, \qquad\quad \bS_A
:=\frac 1 2 \W_{A434}, \qquad\,\, \rhoS:=\frac 14 \W_{3434}\\
\aaS_{AB}:&= \W_{A3B3}, \qquad  \bbS_{A} :=\frac 1 2 \W_{A334}, \qquad \rhodS:=\frac 1 4  \dual \W _{3434}
\end{split}
\eea
We define the spacetime null electromagnetic components of the electromagnetic tensor $\F$ in the following way\footnote{Note that we define the extreme components of the electromagnetic tensor using the $\bF$ notation, as opposed to the standard $\aF$ for the spin $\pm1$ Teukolsky equation. This choice is meant to stress the fact that in electrovacuum background under gravitational and electromagnetic perturbation, the extreme components $\F_{A4}$ and $\F_{A3}$ are not gauge invariant (as the $\b$ component of the curvature is not invariant, as opposed to the extreme component $\a$). See Remark \ref{betanotinvariant}.}:
\bea
\begin{split}
\bF_{A}:&=\F_{A4}, \qquad\quad \bbF_{A}:= \F_{A3}, \quad \rhoF:=& \frac 1 2 \F_{34}, \qquad\quad \rhodF:=\frac 1 2 \dual \F_{34} 
\end{split}
\eea
The Ricci tensor can be expressed in terms of the electromagnetic null decomposition according to Einstein equation \ref{Einsteineq}, and using the decomposition of the Riemann curvature in Weyl curvature and Ricci tensor, 
\beaa
\label{WeylRiemanngeneral}
\R_{\a\b\gamma\delta}=\W_{\a\b\gamma\delta}+\frac 1 2 (\g_{\b\delta}\R_{\a\gamma}+\g_{\a\gamma}\R_{\b\delta}-\g_{\b\gamma}\R_{\a\delta}-\g_{\a\delta}\R_{\b\gamma}),
\eeaa
we can express the full Riemann tensor of the perturbed spacetime in terms of the above decompositions.

Suppose now that the orthonormal   frame on $S$ is adapted to the axial symmetry as follows:  $e_1=e_\vphi=X^{-1/2}  \Z$  with  $X:=\g(\Z, \Z)$, and  $e_2=e_\th$. We define $e_3, e_4, e_\th$ to be the  reduced   null frame, associated to the reduced metric $g$. 

We can define the reduced Ricci coefficients as follows:
\bea
\label{eq-red:Ricci-coeff}
\begin{split}
\chi:&=\chiS_{\th\th},  \quad  \chib :=\chibS_{\th\th}, \quad \eta :=\etaS_\th,\quad \etab :=\etabS_\th, \quad \xi :=\xiS_\th, \quad \xib :=\xibS_\th,\\
\ze :&=\zeS_\th, \quad \om:=\omS, \quad \omb:=\ombS
\end{split}
\eea

We define the reduced curvature and electromagnetic tensor as
\bea \label{reduced}
\begin{split}
\a:&= \aS_{\th\th}, \quad \aa:=\aaS_{\th\th}, \quad \b:=\bS_{\th}, \quad  \bb:=- \bbS_{\th} \quad  \rho:=\rhoS\\
\bF :&= \bF_{\th}, \qquad \bbF:= \bbF_{\th}
\end{split}
\eea

Notice that the polarization implies that every $\Z$-invariant and $\Z$-polarized spacetime tensor $U$ is such that its contraction with an odd number of $e_\vphi=X^{-\frac 1 2 } \Z$ vanishes identically.
Therefore, the polarization of the metric and the $\Z$-invariance of $\F$ imply that the remaining components of the Ricci coefficients, the curvature and the electromagnetic tensor are determined in the following way:
\beaa
\chiS_{\th\vphi}&=&\chibS_{\th\vphi}=\etaS_\vphi=\etabS_\vphi=\xiS_\vphi=\xibS_\vphi=\zeS_\vphi=0, \\
 \chiS_{\vphi\vphi}&=&e_4(\Phi), \quad \chibS_{\vphi\vphi}=e_3(\Phi), \\
\aS_{\th\vphi}&=&\aaS_{\th\vphi}=\bS_{\vphi}=\bbS_{\vphi}=0, \\
\aS_{\vphi\vphi}&=&-\a, \quad \aaS_{\vphi\vphi}=-\aa
\eeaa
and
\bea \label{polarization}
\rhod=0, \qquad \bF_{\vphi}=\bbF_{\vphi}=\rhodF=0
\eea
\begin{remark}
Since in Kerr-Newman spacetime, the components $\rho$ and $\rhodF$ are different from zero for non-zero angular momentum, we see from \eqref{polarization} that the hypothesis of polarization of the metric forces the final state of the evolution to be a non-rotating charged black hole.
\end{remark}

We introduce the notation,
\beaa
\vth:&=&\chi-e_4( \Phi),\qquad   \ka: =\trchS =\chi+e_4 (\Phi )\\
\vthb:&=&\chib-e_3 (\Phi), \qquad   \kab: =\trchbS =\chib+e_3( \Phi )
\eeaa
Thus,
\beaa
\chihS_{\th\th}= -\chihS_{\vphi\vphi}  &=&\frac 1 2 \vth, \qquad \chibhS_{\th\th}   =-\chibhS_{\vphi\vphi}          =\frac 1 2 \vthb
\eeaa
where $\chihS$ and $\chibhS$ are the traceless part of $\chiS$ and $\chibS$ respectively.

\subsection{Reissner-Nordstr{\"o}m spacetime} 

The   Reissner-Nordstr{\"o}m metric has the axial symmetric vector field $\Z=\pr_\vphi$ and the polarized form \eqref{formmetric} of the metric in standard coordinates is given by 
 \beaa
 \g=X^2  d\vphi^2-\Up  dt^2 + \Up ^{-1} dr^2 + r^2 d\th^2, \qquad \Up:= 1-\frac{2m}{r} + \frac{Q^2}{r^2}, \qquad   X= r^2 \sin^2 \th
 \eeaa

 and the electromagnetic tensor is given by 
 \beaa
 \F=-\frac{Q}{r^2} dr \wedge dt.
 \eeaa

\begin{proposition}\label{Reiss} All curvature and electromagnetic components of the Reissner-Nordstr{\"o}m spacetime vanish identically  except 
 \beaa
 \rhoS = -\frac{2m}{r^3} + \frac{2Q^2}{r^4}, \qquad  \rhoF= \frac{Q}{r^2}
\eeaa
\end{proposition}

\section{Main equations in electrovacuum}
Following \cite{Ch-Kl} and \cite{stabilitySchwarzschild}, we derive the main equations in electrovacuum spacetimes and then obtain their reduction,  i.e. their evaluation along $e_\th$, for axially symmetric polarized spacetimes.

\subsection{Null structure equations}

The spacetime\footnote{For convenience  we  drop  the    $^{(1+3)} $ labels in what follows.} null structure equations are 
\beaa
\nabb_3 \chib_{AB} &=&2 \nabb_B \xib_A        -2\omb \chib_{AB} - \chib_{A}^C \chi_{CB} +2\eta_B \xib_A+2\etab_A \xib_B-4\ze_B \xib_A+\R_{A33B}, \\
\nabb_4\chib_{AB}  &=&  2\nabb_B \etab_A +2 \om \chib_{AB} - \chi_B^C \chib_{AC}  +2 (\xi_B\xib_A + \etab_B \etab_A) + \R_{A34B}, \\
\nabb_3 \ze_A&=&-2 \nabb_A \omb -\chib_A^B(\ze_B+\eta_B)+2\omb(\ze_A-\eta_A)+\chi_A^B \xib_B+2\om \xib_A - \frac 1 2 \R_{A334}, \\
\nabb_4 \xib -\nabb_3 \etab &=& 4\om \xib+\chibh\c(\etab-\eta)+\frac 12 \trchb (\etab-\eta) -\frac 1 2 \R_{A334}, \\
\nabb_4 \omb+\nabb_3\om&=& 4\om\omb +\xi\c\xib+\ze\c(\eta-\etab) -\eta\c\etab+\frac 1 4 \R_{3434}, \\
\nabb_C \chib_{AB} + \ze_B \chib_{AC} &=& \nabb_B \chib_{AC} + \ze_C \chib_{AB} + \R_{A3CB}, \\
\gS^{AC} \gS^{BD} \R_{ADCB} &=& 2K + \frac 1 2 \tr\chi \tr\chib - \chih \c \chibh
\eeaa

The symmetric traceless part of the first equation in the reduce picture becomes 
\beaa
e_3(\vthb)+\kab \, \vthb =-2\dds_2\xib - 2\omb \, \vthb +2(\eta+\etab-2\ze)\, \xib-2\aa
\eeaa
where $\dds_2 \xib= -e_\th \xib+ e_\th \Phi \xib$, while its trace gives  
\beaa
e_3(\kab)+\frac 12 \kab^2 +2 \omb \,\kab &=&2\ddd_1\xib+2(\eta+\etab-2\ze)  \xib- \frac 1 2 \vthb\, \vthb - 2\bbF^2
\eeaa
where $\ddd_1\xib=  e_\th\xib+ e_\th(\Phi) \xib$ . 

The symmetric traceless part of the second equation gives the reduced equation
\beaa
e_4\vthb +\frac 12 \ka\, \vthb - 2\om \vthb &=&-2\dds_2\etab-\frac 12 \kab \,\vth+2(\xi\, \xib+\etab^2) - \bF \bbF
\eeaa
and its trace gives
\beaa
e_4(\kab)+\frac 1 2 \ka\, \kab -2\om \kab &=& 2\ddd_1\etab -\frac 1 2  \vth\, \vthb +2(\xi\, \xib+\etab\, \etab)+ 2\rho 
\eeaa
The reduction of the third, fourth and fifth spacetime equation become  
\beaa
e_3 \ze +\frac 1  2 \kab (\ze+\eta) -2\omb(\ze-\eta)&=&\bb -2 e_\th(\omb)+2\om \xib+\frac 12 \ka \,\xib -\frac  1 2 \vthb(\ze+\eta)+\frac 1 2 \vth \,\xib - \rhoF \bbF, \\
e_4(\xib) - e_3(\etab)&=&\bb + 4\om \xib+\frac 12 \kab  (\etab-\eta)+\frac 1 2 \vthb (\etab-\eta)-\rhoF\bbF, \\
 e_4 \omb+e_3\om&=&\rho+ \rhoF^2 + 4\om\omb +\xi\,\xib+\ze(\eta-\etab) -\eta\, \etab
\eeaa
The last two equations are Codazzi and Gauss equations, which in the reduction become
\beaa
\ddd_2 \vthb &=&e_\th(  \kab) -\kab  \ze+ \vthb  \ze - 2\bb -2\rhoF \bbF, \\
K&=& -\frac 1 4 \ka \kab +\frac 1 4 \vth\vthb -\rho  + \rhoF^2
\eeaa
where $\ddd_2 \vthb=e_\th(\vthb)+2 e_\th(\Phi)\vthb$. By the symmetry $e_3-e_4$ we derive the specular equations.

\subsection{Bianchi identities}
In electrovacuum spacetimes, the Bianchi identities for the Weyl curvature have non-homogeneous terms, in particular, using the notations in \cite{Ch-Kl}, we have
\beaa
 \D^\a \W_{\a\b\gamma\delta}&=&\frac 1 2 (\D_\gamma \R_{\b\delta}-\D_{\delta}\R_{\b\gamma})=:J_{\b\gamma\delta} \\
 \D_{[\sigma }\W_{\gamma\delta] \a\b}&=&\g_{\delta \b}J_{\a\gamma\sigma}+\g_{\gamma \a}J_{\b\delta\sigma}+\g_{\sigma \b}J_{\a\delta\gamma}+\g_{\delta \a}J_{\b\sigma\gamma}+\g_{\gamma \b}J_{\a\sigma\delta}+\g_{\sigma \a}J_{\b\gamma\delta}:= \tilde{J}_{\sigma\gamma\delta\a\b}
 \eeaa
 The non-homogeneous terms $J_{\b\gamma\delta}$ and $\tilde{J}_{\sigma\gamma\delta\a\b}$ can be expressed in terms of the electromagnetic components $\bF, \bbF, \rhoF$ and their derivative.
 
The spacetime\footnote{For convenience  we  drop  the    $^{(1+3)} $ labels in what follows.} Bianchi identities equations are
\beaa
\nabb_3\a_{AB}+\frac 1 2 \trchb\,\a_{AB}&=&-2 (\DDs_2\, \b)_{AB}+4\omb \a_{AB} -3(\chih_{AB} \rho+\dual\, \chih_{AB} \rhod)+((\ze+4\eta)\otimes \b)_{AB} +\\
&&+\frac 1 2 (\tilde{J}_{3A4B4}+\tilde{J}_{3B4A4}+J_{434}\delta_{AB}), \\
\nabb_4 \b_A+ 2\trch  \b_A&=&\divv\a_A -2\om \b_A +((2\ze+\etab)\c \a)_A +3(\xi_A \rho +\dual\xi_A \, \rhod)-J_{4A4}, \\
\nabb_3 \b_A+\trchb \b_A &=&\DDs_1(-\rho, \rhod)_A+2(\chih \c\bb)_A +2\omb\, \b_A +(\xib\c \a)_A+3(\eta_A \rho +\dual\eta_A\,   \rhod) + J_{3A4}, \\
\D_4 \rho+\frac 3 2 \trch \rho&=&\divv\b -\frac 1 2 \chibh \c \a +\ze\c\b +2(\etab \c\b-\xi\c\bb) -\frac 1 2 J_{434}
\eeaa
These equations in the reduction are
\bea \label{Bianchicomplete}
\begin{split}
e_3 (\a)+\left(\frac 1 2 \kab -4\omb \right)\a&=&-\dds_2\b - \frac  3 2 \vth \rho +(\ze+4\eta) \b +(2\ze +3\etab+2\eta) \rhoF \aF_{\th}-\xi \rhoF \aaF_{\th}+\\
 &&+e_\th(2\rhoF \bF)+e_4(\bbF\bF )+\chi \bbF\bF-(\chib+2\omb) \bF^2+\\
 &&-  \frac 1 2 e_3(\bF^2) -\frac 1 2  e_4(\rhoF^2 ) -2\chi \rhoF^2, \\
 e_4(\b) +2\left(\ka+\om\right) \b &=&\ddd_2\a+ (2\ze+\etab) \a+ 3\xi \rho -e_{\th}(\bF^2) +e_4(\rhoF \bF) -(2\ze+\etab) (\bF^2)+\\
 &&+2(\om +\chi) \rhoF \bF -2\xi \rhoF^2+\xi \bF \bbF, \\
 e_3 (\b)+\left( \kab-2\omb \right)\b &=&e_\th(\rho)  + 3\eta \rho-\vth\bb +\xib \a +e_{\th}(\rhoF^2) +e_4( \rhoF \bbF)+2\etab\rhoF^2  +\\
&&+(\chi-2\om) \rhoF \bbF -2\omb \bF^2 -\chib\rhoF \bF -\etab \bF \bbF+\xi \bbF^2, \\
e_4 \rho+\frac 3 2 \ka \rho&=&\ddd_1\b -\frac 1 2 \vthb \, \a +\ze\, \b +2(\etab \,\b+ \xi\,\bb ) -\frac 1 2 e_3(\bF^2)+\frac 1 2 e_4(\rhoF^2 )+\\
&&+2\omb\bF^2 + (2\eta-\etab) \rhoF \bF+\xi \rhoF \bbF
\end{split}
\eea
All other equations can be obtained by symmetry $e_3-e_4$.

\subsection{Maxwell's equations}

We write Maxwell's equations \eqref{Max} in null decomposition. The spacetime Maxwell's equations $\D_\a \F_{\b\gamma}+\D_\b \F_{\gamma\a} + \D_\gamma \F_{\a\b}=0$ read
\beaa
\nabb_3 \bF_A+\frac 1 2 \trchb \bF_A&=&\nabb_4 \bbF_A+\frac 1 2 \trch \bbF_A+2\omb\bF_A-2\om\bbF_A+2\DDs_1\rhoF+2(\eta_A+\etab_A)\rhoF
\eeaa

The spacetime Maxwell equations $\D^\a \F_{\a\b}=0$ in null decomposition are
\beaa
\nabb_3\bF_A+\nabb_4\bbF_A&=& -(\frac 1 2 \trch-2\om)\bbF_A -(\chih\c\bbF)_A-(\frac 1 2 \trchb-2\omb)\bF_A -(\chibh\c\bF)_A+(\eta_A-\etab_A)\rhoF, \\
\nabb_3 \rhoF+ \trchb \rhoF&=& -\divv\bbF+(\ze-\eta)\c\bbF+\xib \c \bF, \\
\nabb_4 \rhoF+ \trch \rhoF&=& \divv\bF+(\ze+\etab)\c\bbF-\xi \c \bbF
\eeaa

The reduced equations are therefore
\bea \label{Maxreduced}
\begin{split}
e_4 \bbF&=-e_\th \rhoF-2 \etab\rhoF  +\left(-\frac 1 2 \ka+2\om\right) \bbF+ \frac 1 2 \vthb \bF, \\
e_3 \bF&=e_\th \rhoF+2 \eta\rhoF  +\left(-\frac 1 2 \kab+2\omb\right) \bF+ \frac 1 2 \vth \bbF, \\
e_3 \rhoF+ \ddd_1\bbF&= - \kab \rhoF + (\ze -\eta) \bbF+\xib \bF, \\
-e_4 \rhoF+ \ddd_1\bF&=  \ka \rhoF + (-\ze -\etab) \bF+\xi \bbF
\end{split}
\eea

\section{Quasi-local mass and charge}
Given a    $\Z$-invariant polarized surface $S$ we define    its   volume radius    by the formula
\beaa
|S|=4\pi r^2
\eeaa
where $|S|$ is the volume of the surface using the volume form of the   metric $\gS$.

\begin{definition}\label{charge} We define the quasi-local charge $e=e(S)$ of the foliated spacetime as 
\beaa
e= \frac{1}{4\pi} \int_{S} \rhoF
\eeaa
\end{definition}

Recall that the standard Hawking mass $m_H=m_H(S)$ is defined by 
\beaa
\frac{2m_H}{r}=1+\frac{1}{16\pi}\int_{S_{}}\ka \kab
\eeaa
but in Reissner-Nordstr{\"o}m we have
\beaa
e&=& \frac{1}{4\pi} \int_{S} \rhoF= \frac{1}{4\pi} \int_{S} \frac{Q}{r^2}=Q, \\
\frac{2m_H}{r}&=&1+\frac{1}{16\pi} \int_{S} \ka \kab= 1-(1-\frac{2m}{r}+\frac{Q^2}{r^2}) =\frac{2m}{r}-\frac{Q^2}{r^2}
\eeaa
so the Hawking mass doesn't correspond to the usual mass. We need therefore a definition of a modified quasi-local mass.

\begin{definition}
We define\footnote{Observe that this corresponds to the standard definition of modified Hawking mass $\varpi$ in spherical symmetry.} the modified Hawking  mass  $\varpi=\varpi(S)$  of the foliated spacetime as
  \beaa
\varpi&=& m_H+\frac{e^2}{2r},
\eeaa
\end{definition}
Observe that in Reissner-Nordstr{\"o}m we have
 \beaa
\frac{2\varpi}{r}=\frac{2m}{r}-\frac{Q^2}{r^2}+\frac{Q^2}{r^2}=\frac{2m}{r}
\eeaa
as desired.

Following \cite{stabilitySchwarzschild}, we define average quantities and the difference between a quantity and its average in the following way. Given a function $f$ on $S$ we denote 
\bea \label{def-S-average}
\bar{f}:&=&\frac{1}{|S|} \int_S f, \qquad \check{f}:= f-\bar{f}.
\eea
Observe that, by Definition \ref{charge} of the quasi-local charge $e$,
\bea\label{rhoFcharge}
\overline{\rhoF}&=&\frac{e}{r^2}, \qquad \rhoF=\frac{e}{r^2} +\check{\rhoF}, 
\eea

\section{Perturbations of Reissner-Nordstr{\"o}m spacetime and invariant quantities}
We recall that by Proposition \ref{Reiss}, in Reissner-Nordstr{\"o}m spacetime the Ricci  coefficients  $\xi, \xib, \vth, \vthb, \eta, \etab, \ze$, the curvature
 components $\a,\aa,\b, \bb$ and the electromagnetic components $\bF, \bbF$ vanish identically.   Thus,   roughly, we expect  that  in perturbations of Reissner-Nordstr{\"o}m     these  quantities stay   small, i.e. of order $O(\ep)$ for a sufficiently  small $\ep$. Moreover, recall that under axially symmetric polarized perturbations, we know that $\rhod, \bF_\vphi, \bbF_\vphi, \rhodF=0$, as derived in \eqref{polarization}.       
\begin{definition} 
\label{def:ep-perturbations}
  We say that a       smooth,  electrovacuum,  $\Z$-invariant, polarized   spacetime  is an $O(\ep)$-perturbation of Reissner-Nordstr{\"o}m  if   the following are true:
\bea
\xi, \xib, \vth, \vthb, \eta, \etab, \ze, \qquad \a,\aa,\b, \bb, \qquad \bF, \bbF=O(\ep),
\eea
Also,
\bea
\ka-\ov{\ka}, \quad \kab-\ov{\kab}, \quad \om-\ov{\om},  \quad    \omb-\ov{\omb}, \quad \rho-\ov{\rho} , \quad  \rhoF - \ov{\rhoF}   =O(\ep)
\eea
where $\ov{\ka}, \ov{\kab}, \ov{\om}, \ov{\omb}, \ov{\rho} $    denote       spacetime   averages   as defined in equation \eqref{def-S-average}.    Moreover  the modified Hawking mass and the quasi-local charge are nearly constant,  i.e. 
\bea
d\varpi =O(\ep^2), \qquad de=O(\ep^2)
\eea
Finally,
\beaa
e_3(r)=\frac r 2 \ov{\kab}+O(\ep), \qquad e_4(r)=\frac r 2 \ov{\ka}+O(\ep),\qquad e_\th(r)=0.
\eeaa
\end{definition}

We summarize the linear terms of the null structure equations, the Bianchi identities and the Maxwell equations for perturbations of Reissner-Nordstr{\"o}m in the following proposition. Remark that we used Maxwell equations to simplify the Bianchi identities.

\begin{proposition}\label{Bianchi} Modulo $O(\ep^2)$, the null structure equations, Bianchi identities \eqref{Bianchicomplete} and Maxwell's equations \eqref{Maxreduced} are
\beaa
e_3(\vthb)+\kab \, \vthb &=&2( e_\th(\xib) - e_\th(\Phi)  \xib) - 2\omb \, \vthb -2\aa\\
e_3(\kab)+\frac 12 \kab^2 +2 \omb \,\kab &=&2( e_\th\xib+ e_\th(\Phi) \xib)\\
e_4\vthb +\frac 12 \ka\, \vthb - 2\om \vthb &=&2( e_\th\etab  - e_\th(\Phi) \etab )-\frac 12 \kab \,\vth\\
e_4(\kab)+\frac 1 2 \ka\, \kab -2\om \kab &=& 2( e_\th \etab +   e_\th(\Phi) \etab) + 2\rho  \\
e_3 \ze +\frac 1  2 \kab (\ze+\eta) -2\omb(\ze-\eta)&=&\bb -2 e_\th(\omb)+2\om \xib+\frac 12 \ka \,\xib -\rhoF \bbF\\
e_4(\xib) - e_3(\etab)&=&\bb + 4\om \xib+\frac 12 \kab  (\etab-\eta)-\rhoF \bbF\\
e_4 \omb+e_3\om&=&\rho+ \rhoF^2 + 4\om\omb \\
e_\th(\vthb)+2 e_\th(\Phi)\vthb &=&-2\bb+   (e_\th(  \kab) -\ze \kab  )-2\rhoF \bbF
\eeaa
\beaa 
K&=&-\rho -\frac 1 4 \ka\,\kab+ \rhoF^2\\
e_3 (\a)+\frac 1 2 \kab (\a)&=&(e_\th(\b)- (e_\th \Phi)\b)+ 4\omb (\a) - \frac  3 2 \vth \rho +\rhoF (e_\th \bF-e_\th\Phi \bF-\vth \rhoF) \\
e_4(\b) +2\ka \b&=&( e_\th \a+2e_\th \Phi\a)- 2\om \b+ 3\xi \rho+\rhoF(e_4\bF+2\om\bF-2\xi\rhoF) \\
e_3 (\b)+ \kab \b &=&e_\th(\rho)+2\omb \b   + 3\eta \rho+\rhoF(e_\th \rhoF-\ka\bbF-\frac \kab 2 \bF)\\
e_4 \rho+\frac 3 2 \ka \rho&=&(e_\th(\b)  +(e_\th \Phi)\b)+\rhoF(-\ka \rhoF+e_{\th} \bF +e_\th(\Phi) \bF)\\
e_4 \bbF&=&-e_\th \rhoF-2 \etab\rhoF  +\left(-\frac 1 2 \ka+2\om\right) \bbF \\
e_3 \bF&=&e_\th \rhoF+2 \eta\rhoF  +\left(-\frac 1 2 \kab+2\omb\right) \bF, \\
e_3 \rhoF&=& - \kab \rhoF-\ddd_1\bbF , \\
e_4 \rhoF&=& - \ka \rhoF+\ddd_1\bF
\eeaa
and the other equations are obtained through the symmetry $e_3-e_4$.
\end{proposition}

Our definition of  $O(\ep)$-Reissner-Nordstr{\"o}m perturbations   does not specify a particular  frame. In what follows we investigate  how 
  the main Ricci and  curvature    quantities   change    relative to  frame transformations, i.e linear transformations of the form  $e_\a'= \Om_\a\,^\b  e_\b $
   which take  null frames  into null frames. We will use the fact that a general frame transformation  can be decomposed 
 into the following  three elementary   types:
 \begin{itemize}
 \item   Transformations  which        fix  $e_3$,
  \bea
\label{SSMe:general.e3}
e_3'=e_3, \qquad 
e_\th'=e_\th +\frac 1 2 f e_3,\qquad 
e_4'=  e_4+ f e_\th +\frac 1 4 f^2 e_3
\eea
\item Transformations which fix $e_4$, 
 \bea
\label{SSMe:general.e4}
e_3'=(e_3 + \fb e_\th +\frac 1 4 \fb^2  e_4),\qquad 
e_\th'=e_\th  + \frac 1 2 \fb e_4,\qquad 
e_4'=  e_4 
\eea
 \item  Transformation which preserve the  directions of $e_3, e_4$, i.e  confomal transformations
 of the form  $e_3'= \la  e_3, \, e_4'=\la^{-1} e_4$.   
 \end{itemize}   
where  $f, \fb$ are     reduced $1$-forms and $\la$   is a reduced scalar. A transformation  consistent  with $O(\ep)$-perturbations of Reissner-Nordstrom spacetimes must have $f, \fb =O(\ep)$ and $a:=\log \la =O(\ep)$.   

 \begin{lemma}[Lemma 2.3.1. of \cite{stabilitySchwarzschild}]
 \label{lemma:SSMe:general.composite}
 A general composite transformation  $\mbox{ type}(3)\circ \mbox{ type}(1)\circ\mbox{ type}(2)$  has the form,
  \bea
\label{SSMe:general.composite}
\begin{split}
e_3'&=\la(e_3 + \fb e_\th +\frac 1 4 \fb^2  e_4)\\
e_\th'&=(1+\frac 1 2   f \fb) e_\th  + \frac 1 2  f e_3+\frac 1 2 (\fb+ \frac 1 4 f \fb^2) e_4\\
e_4'&= \la^{-1} \left( (1+\frac 12  f \fb +\frac{1}{16} f^2 \fb^2)  e_4 +(f+\frac 1 4 f^2 \fb) e_\th + \frac 1 4 f^2  e_3\right)
\end{split}
\eea
  \end{lemma}

\begin{proposition}
\label{prop:transformations1}
Under a general transformation of  type \eqref{SSMe:general.composite}  the curvature and electromagnetic components  transform as follows: 
\beaa
\a'&=& \a +O(\ep^2), \qquad \aa'= \aa +O(\ep^2)\\
\b'&=&\la^{-1}\left(\b+\frac 3 2 \rho f \right)   +O(\ep^2),  \qquad \bb'=\la   \left(\bb+\frac 3 2 \rho \fb  \right)   +O(\ep^2), \\
 \bF'&=&\la^{-1}\left(\bF +f \rhoF\right) +O(\ep^2), \qquad \bbF'=\la \left( \bbF - \fb \rhoF\right) +O(\ep^2), \\
  \rho'&=& \rho  + O(\ep^2) , \qquad \rhoF'=  \rhoF+   O(\ep^2)
 \eeaa 
\end{proposition}
\begin{proof} Straighforward calculations using the definitions \eqref{reduced} and Lemma \ref{lemma:SSMe:general.composite}. See also Proposition 2.3.4 of \cite{stabilitySchwarzschild}.
\end{proof}
Notice that the only quantities which vanish in the background and  which are  $O(\ep^2)$  invariant 
  are  the extreme curvature  components   $\a, \aa$. These components verify the Teukolsky equation, which is the first step in the deriving the Regge-Wheeler type equation for the curvature term $\qf$.

\begin{remark}\label{betanotinvariant}
As a consequence of Proposition \ref{prop:transformations1}, the extreme components of the electromagnetic tensor $\bF, \bbF$ are not $O(\ep^2)$ invariant. Moreover, they transform under change of frame similarly to the $\b, \bb$ component of the curvature. This motivates the notation.
 \end{remark}

 \subsection{The new invariant quantity $\ff$}\label{defofff}

To consider the non-linear electromagnetic perturbation of Reissner-Nordstr{\"o}m we need an equation for a $O(\ep^2)$-invariant quantity. By Remark \ref{betanotinvariant}, we can't make direct use of the spin $\pm1$ Teukolsky equation verified by $\bF$ and $\bbF$ (see for example \cite{Federico}), as compared to the electromagnetic perturbation of Schwarzschild, treated in Section \ref{Sch}.  
 We will make use instead of a new quantity for the electromagnetic part of the curvature. From the Bianchi identity for $e_3\a$, we identify the following quantity $\ff$ defined by
\beaa
\ff:&=&\dds_2\bF+\vth\rhoF= -e_\th \bF+e_\th\Phi \bF+\vth \rhoF
\eeaa
This new quantity turns out to play a fundamental role in the equations governing the coupled gravitational and electromagnetic perturbations. Indeed, it appears in the Teukolsky equation for the extreme curvature component $\a$ in electrovacuum. Moreover, there exists a Chandrasekhar-type transformation which transforms $\ff$ into $\qf^{\F}$, and $\qf^{\F}$ verifies a Regge-Wheeler type equation coupled with the curvature as in \eqref{schemesystem}. And most importantly, $\ff$ is a $O(\ep^2)$ invariant quantity.

\begin{lemma} The quantity $\ff$ is $O(\ep^2)$ invariant.
\end{lemma}
\begin{proof} Using Lemma \ref{lemma:SSMe:general.composite} together with the definition of $\vth$, we have  that  $\vth'=\vth +  e_\th(f)  - f e_\th(\Phi) +O(\ep^2)$. Using Proposition \ref{prop:transformations1} and that $e_\th\rhoF=O(\ep)$ as a consequence of Maxwell equations, we have
\beaa
\ff'&=& -e_\th' \bF'+e_\th'\Phi \bF'+\vth' \rhoF'=\\
&=&-e_\th (\bF +f \rhoF)+e_\th\Phi (\bF +f \rhoF)+(\vth +  e_\th(f)  - f e_\th(\Phi)) \rhoF+O(\ep^2)=\\
&=&\ff-e_\th (f) \rhoF+e_\th\Phi (f \rhoF)+(e_\th(f)  - f e_\th(\Phi) ) \rhoF+O(\ep^2)=\ff+O(\ep^2)
\eeaa
therefore $\ff$ is $O(\ep^2)$ invariant.
\end{proof}

\begin{remark}\label{remChandra} To the knowledge of the author, the quantity $\ff$ seems to not have been noticed or used so far in the literature. One main reason could be found in the choice of gauge of Chandrasekhar in \cite{Chandra}. In treating the gravitational and electromagnetic perturbation of Reissner-Nordstr{\"o}m, the author picks the phantom gauge $\phi_0=0$, corresponding to $\bF=\bbF=0$. In fact, from Proposition \ref{prop:transformations1}, we can see that the choice of gauge with $f=-\rhoF^{-1} \bF$ and $\bar{f}=-\rhoF^{-1} \bbF$ gives $\bF, \bbF=O(\ep^2)$ in the non-linear setting. This choice of gauge would reduce $\ff$ to $\vth \rhoF$.
\end{remark}

Applying Maxwell's equations and null structure equations, and using that $[e_\th, e_3]=\frac 1 2  \kab  e_\th +O(\ep)$ and $[e_\th, e_4]=\frac  1 2 \ka e_\th +O(\ep) $ we write
\bea\label{e3f}
\begin{split}
e_3(\ff)&=& -(e_\th e_3 \bF-e_\th\Phi e_3\bF)+\frac 1 2 \kab (e_\th \bF- e_\th\Phi \bF)-\rhoF\Big((\frac 32 \kab-2\omb)\, \vth  -2( e_\th\eta  - e_\th(\Phi) \eta )+\frac 12 \ka \,\vthb\Big) , \\
e_4(\ff)&=&-(e_\th e_4 \bF-e_\th\Phi e_4\bF)+\frac 1 2\ka (e_\th \bF- e_\th\Phi \bF)-\rhoF \Big(2\ka \, \vth -2( e_\th(\xi) - e_\th(\Phi)  \xi) + 2\om \, \vth +2\a\Big) 
\end{split}
\eea
to be used later. 

\section{Teukolsky equations for the $O(\ep^2)$-invariant quantities $\a$ and $\ff$}
Let $f$ be a $\Z$-invariant scalar function. Then, by definition $\square_\g$ for a polarized metric $\g$, we have
\bea \label{squaref}
\square_\g f &=& -e_4(e_3( f))+e_\th(e_\th( f)) -\frac{1}{2}\kab e_4(f) +\left(-\frac{1}{2}\ka+2\om\right) e_3(f)+e_\th(\Phi)e_\th( f) +2\etab e_\th(f).
\eea
as in Lemma 2.4.1. of \cite{stabilitySchwarzschild}. Using this formula, we derive the wave equations for the invariant quantities $\a$ and $\ff$.

\begin{proposition}\label{squarea}[Teukolsky equation for $\a$] The $O(\ep^2)$ invariant quantity $\a$ verifies the following wave equation:
\beaa
\square_\g\a &=& -4\omb e_4(\a) +(2\ka+4\om)e_3 (\a)+\Big(\frac 1 2 \ka\kab+2\om \kab- 4\rho+4\rhoF^2-4e_4\omb -10\ka\omb-8\om\omb+ 4e_\th(\Phi)^2 \Big)\a+\\
&&+\rhoF\Big(2 e_4(\ff)+( 2\ka+4\om) \ff \Big)  +O(\ep^2)   
\eeaa
\end{proposition}
\begin{proof}
Using the Bianchi identity
\beaa
e_3 (\a)+\frac 1 2 \kab (\a)&=&(e_\th(\b)- (e_\th \Phi)\b)+ 4\omb (\a) - \frac  3 2 \vth \rho -\rhoF \ff
\eeaa
 and applying Bianchi identities, null structure equations and Maxwell's equations as in Proposition \ref{Bianchi} and formulas \eqref{e3f}, we have
\beaa
e_4(e_3(\a)) &=&  e_4(e_\th( \b)) -  e_\th(\Phi) e_4(\b)  -e_4(e_\th(\Phi)) \b  -\left(\frac{\kab}{2} \right)e_4(\a) -\left(\frac{e_4(\kab)}{2}\right)\a - \frac 3 2 \vth e_4(\rho) - \frac 3 2 e_4(\vth) \rho +4e_4\omb\a +4\omb e_4\a+\\
&&-e_4\rhoF \ff-\rhoF e_4\ff\\
&=&  e_4(e_\th( \b)) -  e_\th(\Phi)\left(e_\th(\a)+2e_\th(\Phi)\a - 2(\ka+\om) \b +3\xi\rho+\rhoF(e_4\bF+2\om\bF-2\xi\rhoF)\right)+\\
&& - (-\frac \ka 2 e_\th\Phi)\b  -\left(\frac{\kab}{2}\right)e_4(\a) -\left(\frac{e_4(\kab)}{2}\right)\a- \frac 3 2 \vth e_4(\rho) - \frac 3 2 e_4(\vth) \rho+4e_4\omb\a +4\omb e_4\a+\\
&&+\rhoF \left(\ka \ff+(e_\th e_4 \bF-e_\th\Phi e_4\bF)-\frac 1 2\ka (e_\th \bF- e_\th\Phi \bF)+\rhoF (2\ka \, \vth -2( e_\th(\xi) - e_\th(\Phi)  \xi) + 2\om \, \vth +2\a) \right) \\
&=&  e_4(e_\th( \b)) -  e_\th(\Phi)(e_\th(\a)+2e_\th(\Phi)\a) +2e_\th(\Phi)(\ka+\om) \b  -3  e_\th(\Phi)\xi\rho +\\
&&  + \frac \ka 2 e_\th(\Phi)\b  -\left(\frac{\kab}{2}\right)e_4(\a) -\left(\frac{e_4(\kab)}{2} \right)\a - \frac 3 2 \vth e_4(\rho) - \frac 3 2 e_4(\vth) \rho +4e_4\omb\a +4\omb e_4\a+\\
&&+\rhoF\Big((e_\th e_4 \bF-2e_\th\Phi e_4\bF)-\frac 3 2\ka e_\th \bF+(\frac 3 2\ka-2\om) e_\th\Phi \bF+\\
&&+\rhoF (3\ka \, \vth -2( e_\th(\xi) - 2e_\th(\Phi)  \xi) + 2\om \, \vth +2\a) \Big) 
\eeaa
Using that $e_\th (\ka), e_\th(\om), e_\th(\rho), e_\th (\rhoF)=O(\ep)$, we have
\beaa
e_\th(e_\th(\a)) &=& e_{\th}(e_4(\b) +2(\ka+\om) \b -2e_\th(\Phi)\a -3\xi\rho-\rhoF e_4\bF-2\om\rhoF \bF +2\xi \rhoF^2)\\
&=&  e_\th(e_4(\b)) +2(\ka+\om) e_\th(\b)   - 2e_\th(\Phi)e_\th(\a) - 2e_\th(e_\th(\Phi))\a-3e_\th(\xi)\rho+\rhoF( - e_\th e_4\bF-2\om e_\th \bF +2e_\th \xi \rhoF)\\
&=&   e_\th(e_4(\b))+2(\ka+\om) \Big(e_\th(\Phi)\b + e_3 (\a) +\left(\frac{\kab}{2}\right)\a + \frac 3 2 \vth\,\rho-4\omb\a\Big)  - 2e_\th(\Phi)e_\th(\a)+\\
&& - 2(\rho-\rhoF^2-(e_\th\Phi)^2+\frac 1 4 \ka\kab)\a -3e_\th(\xi)\rho+\\
&&+\rhoF\left(- e_\th e_4\bF+(-2\ka-4\om) e_\th \bF +(2\ka+2\om) e_\th \Phi \bF+\rhoF(2\ka \vth+2\om\vth+2e_\th \xi)\right)
 \eeaa
Using \eqref{squaref}, we have
\beaa
\square_\g\a &=&  -e_4(e_3(\a))+e_\th(e_\th(\a)) -\frac{1}{2}\kab e_4(\a) +\left(-\frac{1}{2}\ka+2\om\right) e_3(\a)+e_\th(\Phi)e_\th(\a)\\
&=& [e_\th, e_4](\b)  +  2e_\th(\Phi)^2\a  +3  e_\th(\Phi)\xi\rho  - \frac \ka 2 e_\th(\Phi)\b  +\left(\frac{\kab}{2}\right)e_4(\a) +\left(\frac{e_4(\kab)}{2} \right)\a + \frac 3 2 \vth e_4(\rho) + \frac 3 2 e_4(\vth) \rho -4e_4\omb\a -4\omb e_4\a+\\
&& +2(\ka+\om) e_3 (\a) +2(\ka+\om) \left(\frac{\kab}{2} \right)\a +3(\ka+\om) \vth\,\rho -8(\ka\omb+\om\omb)\a- 2(\rho-\rhoF^2-(e_\th\Phi)^2+\frac 1 4 \ka\kab)\a -3e_\th(\xi)\rho+\\
&&-\frac{1}{2}\kab e_4(\a) +\left(-\frac{1}{2}\ka+2\om\right) e_3(\a)+\\
&&+\rhoF\Bigg(-2(e_\th e_4 \bF-e_\th\Phi e_4\bF)-(\frac 1 2\ka +4\om)(e_\th \bF- e_\th\Phi \bF)-\rhoF (\ka \, \vth -2( 2e_\th(\xi) - 2e_\th(\Phi)  \xi) +2\a)\Bigg)
\eeaa
Using that modulo $O(\ep^2)$,
\beaa
[e_\th, e_4](\b) &=& \frac \ka 2 e_\th(\b)= \frac \ka 2 \left(  e_\th(\Phi) \b  +   e_3 (\a) +\left(\frac{\kab}{2}\right)\a -4\omb\a + \frac 3 2 \vth\,\rho+\rhoF(\vth\rhoF -e_\th \bF+ e_\th \Phi \bF) \right)
\eeaa 
we have
\beaa
\square_\g\a &=& 3  e_\th(\Phi)\xi\rho -3e_\th(\xi)\rho + \frac 3 2 \vth e_4(\rho) + \frac 3 2 e_4(\vth) \rho +3(\ka+\om) \vth\,\rho+\frac 3 4 \ka\vth\rho\\
&&  +\left(-4\omb \right)e_4(\a) +(2\ka+4\om)e_3 (\a)+\\
&&+\Bigg(\frac{e_4(\kab)}{2}+  4e_\th(\Phi)^2-4e_4\omb+\frac 3 4 \ka\kab+\om \kab -10\ka\omb-8\om\omb- 2\rho+2\rhoF^2 \Bigg)\a+\\
&&+\rhoF\Bigg(-2(e_\th e_4 \bF-e_\th\Phi e_4\bF)-(\ka +4\om)(e_\th \bF- e_\th\Phi \bF)-\rhoF (\frac 1 2 \ka \, \vth -2( 2e_\th(\xi) - 2e_\th(\Phi)  \xi) +2\a) \Bigg)
\eeaa
Using equations for $e_4(\vth), e_4 \rho, e_4(\kab)$ we infer 
\beaa
\square_\g\a &=& -4\omb e_4(\a) +(2\ka+4\om)e_3 (\a)+\left( 4e_\th(\Phi)^2-4e_4\omb+\frac 1 2 \ka\kab+2\om \kab -10\ka\omb-8\om\omb- 4\rho \right)\a+\\
&&+\rhoF\Bigg(-2(e_\th e_4 \bF-e_\th\Phi e_4\bF)-(\ka +4\om)(e_\th \bF- e_\th\Phi \bF)-\rhoF (2 \ka \, \vth -2( 2e_\th(\xi) - 2e_\th(\Phi)  \xi) \Bigg)
\eeaa
Using \eqref{e3f}, we can write the term multiplying $\rhoF$ on the right hand side as $(2\ka +4\om)\ff+2 e_4\ff+4\rhoF\a$ giving therefore the desired expression.
\end{proof}

\begin{proposition}\label{squareff}[Teukolsky equation for $\ff$] The $O(\ep^2)$-invariant quantity $\ff$ verifies the following wave equation:
\beaa
\square_\g \ff&=&\left(\frac 3 2 \ka +2\om\right) e_3 \ff+\left(\frac 1 2 \kab -2\omb\right) e_4\ff+\left(\frac 1 2 \ka\kab+2\om\kab-2\rho-4\ka\omb-2e_4\omb +4e_\th\Phi^2\right) \ff+\\
&&+ \rhoF \Big (-2e_3 (\a)-(2\kab-8\omb)\a \Big) +O(\ep^2)
\eeaa
\end{proposition}
\begin{proof} We derive the Teukolsky equation verified by $\bF$. Consider the Maxwell's equations:
\beaa
e_3 \bF-e_\th \rhoF-2 \eta\rhoF  +\left(\frac 1 2 \kab-2\omb\right) \bF=0, \\
e_4 \rhoF+\ka \rhoF-e_{\th} \bF -e_\th(\Phi) \bF=0
\eeaa
and apply the operator $(e_4 +\frac 3 2 \ka)$ to the first equation and the operator $(e_\th +2\eta)$ to the second and add them, keeping only the linear terms. We are left with:
\beaa
0&=&-(- e_4 e_3 \bF+e_\th e_{\th} \bF+e_\th(\Phi) e_\th \bF-\frac 1 2 \kab e_4\bF+(-\frac 1 2 \ka+2\om) e_3\bF)+\\
&& +(\frac 1 2 \ka\kab+\om\kab+\rho-3\ka\omb -2e_4\omb -e_\th e_\th(\Phi))\aF_\th-2\omb e_4\bF+ (\ka+2\om) e_3 \bF+\\
&&[e_\th, e_4] \rhoF-\frac 1 2 \ka e_\th \rhoF+(-   \ka\eta +e_\th \ka  -2e_4\eta) \rhoF=\\
&=&-\square_\g \bF +(\frac 1 4 \ka\kab+\om\kab-3\ka\omb -2e_4\omb +\rhoF^2+(e_\th\Phi)^2)\bF-2\omb e_4\bF+ (\ka+2\om) e_3 \bF+\\
&&(e_\th(\vth)+2 e_\th(\Phi)\vth  -2 e_3(\xi) +\xi \kab  + 8\omb \xi+4\b) \rhoF
\eeaa
which gives an expression for $\square_\g \bF$. Recall that $\ff=-e_\th \bF+e_\th\Phi \bF+\vth \rhoF$, and using that 
\beaa
[\square_\g, e_\th]f &=& \frac 1 2 \ka e_3 e_\th(f)+\frac 1 2 \kab e_4 e_\th f +(\rhoF^2+(e_\th \Phi)^2+\frac 1 4 \ka\kab)e_\th( f), \\
\square (e_\th \Phi)&=&e_\th \Phi (\rhoF^2+(e_\th \Phi)^2 -\frac 1 4 \ka \kab), \\
\square(\vth)&=&e_\th e_\th \vth +e_\th \Phi e_\th \vth+ \left(-\frac 1 4 \ka\kab+\rho+\omb\ka+8\om\omb-2e_4\omb\right)\, \vth+ \left(-\frac 1 4 \ka^2 -\om\ka\right) \,\vthb +\\
&&-2(e_\th e_3(\xi)- e_\th(\Phi) e_3(\xi))+4\omb (e_\th \xi-e_\th\Phi \xi)+( \ka+4\om)( e_\th\eta  - e_\th(\Phi) \eta )+\\
&&+4\omb \a+2(e_\th\b-e_\th(\Phi)\b)+2\rhoF (e_\th\aF_{\th}- e_\th(\Phi)\aF_{\th}) , \\
\square_\g \rhoF &=& \rhoF(-\frac 1 2 \ka \kab +2\rho)+O(\ep)
\eeaa
which can all be derived as consequences of null structure equations, Bianchi identities and Maxwell's equation through the formula \eqref{squaref}, we have
\beaa
\square (e_\th \bF)&=& e_\th (\square \bF) + [\square, e_\th] \bF=\\
&=& \Big(\frac 1 4 \ka \kab-3\ka\omb+\om\kab +\rhoF^2 +e_\th\Phi^2-2e_4\omb\Big) e_\th\bF+2e_\th\Phi e_\th e_\th \Phi \bF+\left(\ka+2\om\right) e_\th e_3(\bF)-2\omb e_\th e_4\bF+\\
&&+\Big( e_\th e_\th(\vth)+2 e_\th(\Phi)e_\th\vth+2 e_\th e_\th(\Phi)\vth +e_\th\xi \kab-2 e_\th e_3(\xi) +8\omb e_\th\xi+4e_\th \b\Big) \rhoF+\\
&& + \frac 1 2 \ka e_3 e_\th(\bF)+\frac 1 2 \kab e_4 e_\th \bF +(\rhoF^2+(e_\th \Phi)^2+\frac 1 4 \ka\kab)e_\th( \bF), \\
&=& \left(\frac 3 2 \ka +2\om\right) e_\th e_3 \aF+\left(\frac 1 2 \kab -2\omb\right) e_\th e_4 \aF+\left(-3\ka\omb+\om\kab +2\rhoF^2 +2e_\th\Phi^2-2e_4\omb\right) e_\th\aF_{\th}+\\
 &&+ \left(2\rho-2\rhoF^2-2(e_\th\Phi)^2+\frac 12 \ka\kab\right)e_\th\Phi \aF_{\th}+\\
&&+\rhoF\Big( e_\th e_\th(\vth)+2 e_\th(\Phi)e_\th\vth+ (2\rho-2\rhoF^2-2(e_\th\Phi)^2+\frac 12 \ka\kab)\vth+e_\th\xi \kab-2 e_\th e_3(\xi) +8\omb e_\th\xi+4e_\th \b\Big) 
\eeaa
\beaa
\square( e_\th \Phi \bF)&=&\square_\g(e_\th \Phi) \bF +e_\th \Phi \square_\g (\bF)-e_3 e_\th \Phi e_4 \bF-e_4 e_\th \Phi e_3 \bF +2 e_\th e_\th \Phi e_\th \bF=\\
&=&\left(2\rho -2\rhoF^2 -2e_\th \Phi^2 +\frac 1 2 \ka \kab\right) e_\th \bF+e_\th \Phi \Big(\left(\frac 3 2 \ka+2\om\right) e_3(\bF_\th)+\left(\frac 1 2 \kab -2\omb \right)e_4 \bF+\\
&&+ \left(-3\ka\omb-2e_4\omb+2\rhoF^2+2(e_\th \Phi)^2+\om\kab\right) \bF +\Big( e_\th(\vth)+2 e_\th(\Phi)\vth +\xi \kab-2 e_3(\xi)+8\omb\xi +4\b\Big) \rhoF\Big)
\eeaa
Putting all this together we get 
\beaa
\square_\g (\ff)&=& \square_\g(-e_\th \bF+e_\th\Phi \bF+\vth \rhoF)=\\
&=& \square_\g(-e_\th \bF+e_\th\Phi \bF)+\square(\vth)\rhoF +\vth \square(\rhoF)-e_3(\vth) e_4(\rhoF) -e_4(\vth) e_3(\rhoF)=\\
&=&-\left(\frac 3 2 \ka +2\om\right) (e_\th e_3 \bF-e_\th \Phi e_3\bF)-\left(\frac 1 2 \kab -2\omb\right) (e_\th e_4 \bF-e_\th \Phi e_4\bF)+\\
 &&-\left(-\frac 1 2 \ka\kab-3\ka\omb+\om\kab-2\rho +2\rhoF^2 +4e_\th\Phi^2-2e_4\omb\right) (e_\th\bF-e_\th \Phi \bF)+\\
&&-\rhoF\Big(\left(\frac {11}{4} \ka\kab+2\om\kab-\rho-2\rhoF^2-4(e_\th\Phi)^2-3\omb\ka-8\om\omb+2e_4\omb\right)\vth+\left(\frac 3 4 \ka^2 +\om\ka\right) \,\vthb+\\
&& +(-\kab+4\omb) (e_\th\xi-e_\th \Phi\xi)+( -3\ka-4\om)( e_\th\eta  - e_\th(\Phi) \eta )+2(e_\th \b-e_\th \Phi\b)+(2\kab-4\omb)\a\Big)
\eeaa
and using once again \eqref{e3f} and $$(e_\th(\b)- (e_\th \Phi)\b)=e_3 (\a)+\frac 1 2 \kab \a-4\omb\a + \frac  3 2 \vth \rho +\vth \rhoF^2-\rhoF(e_\th \bF-e_\th\Phi \bF)$$ we obtain
\beaa
\square_\g (\ff)&=&\left(\frac 3 2 \ka +2\om\right) e_3 \ff+\left(\frac 1 2 \kab -2\omb\right) e_4\ff-\left(\frac 1 2 \ka\kab-4\ka\omb+2\om\kab-2\rho +4e_\th\Phi^2-2e_4\omb\right) (e_\th\bF-e_\th \Phi \bF)+\\
&&-\rhoF\Big(\left(-\frac {1}{2} \ka\kab-2\om\kab+2\rho-4(e_\th\Phi)^2+4\omb\ka+2e_4\omb\right)\vth+2e_3 (\a)+ (2\kab-8\omb)\a \Big)=\\
&=&\left(\frac 3 2 \ka +2\om\right) e_3 \ff+\left(\frac 1 2 \kab -2\omb\right) e_4\ff+\left(\frac 1 2 \ka\kab+2\om\kab-2\rho-4\ka\omb-2e_4\omb +4e_\th\Phi^2\right) \ff+\\
&&+2 \rhoF \Big (-e_3 (\a)- (\kab-4\omb)\a \Big)
\eeaa
as desired.
\end{proof}

\begin{remark} Observe that equations \eqref{squarea} and \eqref{squareff}, respectively for the Weyl curvature $\a$ and for the electromagnetic tensor $\ff$,
\beaa
\square_\g\a &=& -4\omb e_4(\a) +(2\ka+4\om)e_3 (\a)+\Big(\frac 1 2 \ka\kab+2\om \kab- 4\rho+4\rhoF^2-4e_4\omb -10\ka\omb-8\om\omb+ 4e_\th(\Phi)^2 \Big)\a+\\
&&+\rhoF\Big(2 e_4(\ff)+( 2\ka+4\om) \ff \Big)  +O(\ep^2), \\
\square_\g \ff&=&\left(\frac 1 2 \kab -2\omb\right) e_4\ff+\left(\frac 3 2 \ka +2\om\right) e_3 \ff+\left(\frac 1 2 \ka\kab+2\om\kab-2\rho-4\ka\omb-2e_4\omb +4e_\th\Phi^2\right) \ff+\\
&&+ \rhoF \Big (-2e_3 (\a)-(2\kab-8\omb)\a \Big) +O(\ep^2)
\eeaa 
are coupled. As in \cite{Ch-Kl}, signature arguments apply. The component $\a$ has signature $2$ and the quantity $\ff$ has signature $1$, therefore in the wave equation for $\a$, $\ff$ has to appear with an $e_4$ derivative. On the other hand, in the wave equation for $\ff$ the component $\a$ has to appear with an $e_3$ derivative. Moreover, this coupling comes with a multiplication for $\rhoF$, and recall that from \eqref{rhoFcharge}, $\rhoF$ can be interpreted as a weighted quasi-local charge of the spacetime. 
\end{remark}

\section{System of equations for the coupled gravitational and electromagnetic perturbations}
In Schwarzschild spacetime, the Chandrasekhar's transformation applied to the extreme curvature component $\a$ gives a quantity at the level of the second derivative along the ingoing null direction of $\a$ that verifies a Regge-Wheeler equation (see for example the quantity $P$ in \cite{DHR}, or the quantity $\qf$ in \cite{stabilitySchwarzschild}). In Reissner-Nordstr{\"o}m spacetime, we will get a Regge-Wheeler type equation, i.e. a wave equation with a good potential and no lower order terms, but with additional terms giving the coupling with the electromagnetic tensor. The new result is that there exists a transformation similar to Chandrasekhar's one, at the level of one derivative along the ingoing null direction, that can be applied to the new quantity $\ff$ to obtain a Regge-Wheeler type equation for the electromagnetic term $\qf^{\F}$, with additional terms giving the coupling with the curvature.

Inspired by the system of three equations for the extreme curvature component and its two derivatives in slowly rotating Kerr as obtained in \cite{Siyuan}, we write a system of five equations for suitably chosen rescaled quantities depending on the curvature and on the electromagnetic components, from the two quantities $\a$ and $\ff$.
\begin{itemize}
\item The first three equations are equations for the rescaled $\a$, its first and its second derivative in the ingoing null direction, respectively. The third quantity corresponds to the $\qf$ obtained by Chandrasekhar transformation in \cite{stabilitySchwarzschild} which verifies the Regge-Wheeler type equation with a new right hand side depending on the electromagnetic components. 
\item The last two equations are equations for the rescaled $\ff$ and its first derivative in the ingoing null directions $\qf^{\F}$. This last quantity turns out to verify a Regge-Wheeler type equation too, with a right hand side depending on the curvature.
\end{itemize}

The first three equations correspond to the equations obtained by Ma in \cite{Siyuan} in the case of Kerr spacetime. On the right hand side, his equations have lower order terms in the curvature multiplied by the angular momentum. In the case of small angular moment, the author is able to absorb the error terms coming from the lower error terms. In the case of coupled gravitational and electromagnetic perturbations, the right hand side of the first three equations is not given by lower order terms, but from a non trivial dependence on the electromagnetic parts, which are independent to the curvature part. 

\subsection{Definition of rescaled quantities and operators}\label{acknSteff}
As suggested in \cite{Steffen}, we introduce the following operators as a rescaled version of the derivative in the ingoing and in the outgoing null directions\footnote{To be consistent with the previous definitions, we define the operator as $\underline{P}$, since the bar quantities refer to $e_3$. The operator $Q$ differ from the operator $P=r \ka^{-1} e_4f +\frac 12 rf$ which would be used in the treatment of the corresponding system for the spin $-2$ quantity $\aa$. See Remark \ref{spin-2}.}:
\bea \label{operators}
\underline{P}(f)&=&r \kab^{-1} e_3f +\frac 12 rf, \\
Q(f)&=&r\kab   e_4f+\frac 1 2r \ka\kab f
\eea 
The operator $\underline{P}$ is fundamentally used to define the various quantities in \eqref{quantities}, while $Q$ is introduced to simplify the right hand side of the Teukolsky equation for $\a$. Observe that the operators $\underline{P}$ and $Q$, even if consist in derivatives along the $e_3$ and $e_4$ directions, don't change the signature of the quantity they are applied to.

We compute $\square_\g (\underline{P}(f))$ and $e_3 (Qf)$, which will be useful in the derivation of the main equations.
\begin{lemma}\label{squareP} We have, modulo $O(\ep^2)$,
\beaa
\square_\g(\underline{P} f)&=& \frac{1}{r}( -\ka\kab+2\rho) \underline{P}(\underline{P}(f))+(\frac 1 2 \ka\kab-4\rho-2\rhoF^2)\underline{P}(f)+\left(\frac 1 2 \rho+\rhoF^2 \right)rf+ \frac 3 2r \square_\g(f)+\kab^{-1}r e_3(\square_\g(f))
\eeaa
\end{lemma}
\begin{proof} Writing $\underline{P}f= r \kab^{-1} e_3f +\frac 12 r f$, we have 
\beaa
\square_\g(\underline{P} f)&=&\frac 12 \Big(\square_\g( r)f+r \square_\g(f)-e_3(r)e_4(f)-e_4(r)e_3(f)\Big)+\square_\g(\kab^{-1}r) e_3(f)+\\
&&+\kab^{-1}r \square_\g( e_3(f))-e_3(\kab^{-1}r) e_4e_3(f)-e_4(\kab^{-1}r)e_3 e_3(f)
\eeaa

Using that 
\beaa
\square_\g( r)&=&r( -\frac1 2  \ka\kab  -\rho), \\
e_3(\kab^{-1})&=&-\frac{1}{\kab^2}e_3(\kab)=-\frac{1}{\kab^2} (-\frac 12 \kab^2-2\omb\kab)=\frac 12 +2\omb\kab^{-1}, \\
e_4(\kab^{-1})&=&-\frac{1}{\kab^2}e_4(\kab)=-\frac{1}{\kab^2}(-\frac 1 2 \ka\kab+2\om\kab+2\rho)=\frac 1 2 \ka\kab^{-1}-2\om\kab^{-1}-2\rho\kab^{-2}, \\
\square_\g(\kab)&=&\kab \rho+2\kab e_4\omb+4\omb \rho, \\
\square_\g( e_3(f))&=&e_3(\square_\g(f))+[\square_\g, e_3]f=\\
&=&e_3(\square_\g(f))+\kab\square_\g f - 2\om e_3(e_3( f)) +(\kab+2\omb)e_4(e_3( f))+\\
&&+\left(\frac{1}{4}\ka\kab  - 3\om\kab  +\omb\ka -8\om\omb -\rho-2\rhoF^2  +2e_4(\omb)\right)e_3(f)  +\frac{1}{4}\kab^2 e_4(f)
\eeaa
 we have 
\beaa
\square_\g(\underline{P} f)&=& r \Bigg[( -\ka\kab^{-1}+2\rho\kab^{-2})e_3 e_3(f)+( -\frac14 \ka\kab  -\frac 1 2 \rho)f+ \frac 3 2 \square_\g(f)+\\
&&+\Big(- \frac 3 2 \ka-2\omb\ka\kab^{-1}+4\omb\rho\kab^{-2} -2\rhoF^2\kab^{-1} \Big) e_3(f)+\kab^{-1} e_3(\square_\g(f))\Bigg]
\eeaa
Writing 
\beaa
 e_3(f)&=&\frac{1}{r}\kab (Pf)-\frac 1 2 \kab f, \\
e_3e_3f&=&\frac{1}{r^2}\kab^2 P(P(f))-\frac 2 r(\kab^2+\omb \kab) (Pf)+(\frac 1 2 \kab^2+\omb\kab) f
 \eeaa
 then 
 \beaa
 \square_\g(\underline{P} f)&=& r \Bigg(\frac{1}{r^2}( -\ka\kab+2\rho) P(P(f))+\frac{1}{r}(\frac 1 2 \ka\kab-4\rho-2\rhoF^2)P(f)+\left(\frac 1 2 \rho+\rhoF^2 \right)f+ \frac 3 2 \square_\g(f)+\kab^{-1} e_3(\square_\g(f))\Bigg)
 \eeaa
 as desired.
\end{proof}

\begin{lemma}\label{e3underlineP} We have, modulo $O(\ep^2)$,
\beaa
e_3 (Q(f))&=&\frac 1 r  \kab (Q(\underline{P}(f)))-\frac 1 2 \kab Q(f)+(-\ka\kab^2+2\rho\kab)\underline{P}(f)
\eeaa
\end{lemma}
\begin{proof} We have 
\beaa
e_3 (Qf)&=&e_3(\kab e_4 (rf))=e_3(\kab r e_4(f)+\frac 1 2 \ka \kab rf)=\\
&=&e_3(\kab) r e_4(f)+\kab e_3(r) e_4(f)+\kab r e_3e_4(f)+\frac 1 2 e_3(\ka) \kab rf+\frac 1 2 \ka e_3(\kab) rf+\frac 1 2 \ka \kab e_3(r)f+\frac 1 2 \ka \kab r e_3(f)=\\
&=&\kab r e_4e_3(f)+(\frac 1 2 \ka\kab-2\om \kab )re_3(f)+(-\frac 1 4 \ka \kab^2 +\rho \kab) rf
\eeaa
Writing $e_3(f)=\frac{1}{r}\kab \ \underline{P}(f)-\frac 1 2 \kab f$, we have 
\beaa
e_3 (Qf)&=&\kab r e_4(\frac{1}{r}\kab \ \underline{P}f-\frac 1 2 \kab f)+(\frac 1 2 \ka\kab^2-2\om \kab^2 )\underline{P}(f)+(-\frac 1 2 \ka \kab^2+\om\kab^2 +\rho \kab) rf=\\
&=&\kab r (-\frac {1}{ 2 r} \ka \kab \underline{P}f+\frac{1}{r}(-\frac 1 2 \ka\kab+2\om\kab+2\rho) \underline{P}f+\frac{1}{r}\kab e_4(\underline{P}f)-\frac 1 2 (-\frac 1 2 \ka\kab+2\om\kab+2\rho) f+\\
&&-\frac 1 2 \kab e_4(f))+(\frac 1 2 \ka\kab^2-2\om \kab^2 )\underline{P}f+(-\frac 1 2 \ka \kab^2+\om\kab^2 +\rho \kab) rf
\eeaa
which gives
\beaa
e_3 (Qf)&=&\kab^2 e_4(\underline{P}f)-\frac 1 2 \kab^2 r  e_4(f)+(-\frac 1 2 \ka\kab^2+2\rho\kab)\underline{P}(f)+(-\frac 1 4 \ka \kab^2 ) rf
\eeaa
and writing $ e_4(f)=\frac 1 r \kab^{-1}Qf-\frac 1 2 \ka f$, and $e_4(\underline{P}f)=\frac 1 r \kab^{-1} Q\underline{P}f-\frac 1 2 \ka \underline{P}f$ we have the desired expression
\end{proof}

As suggested in \cite{Steffen}, we define the following new scaling of the extreme components of the curvature $\a$ and of the electromagnetic quantity $\ff$ as the following:
\bea\label{quantities}
\begin{split}
\Phi_0 &= r^2 \kab^2 \a, \\
\Phi_1&=\underline{P}(\Phi_0), \\
\Phi_2&=\underline{P}(\Phi_1)=\underline{P}(\underline{P}(\Phi_0)), \\
\Phi_3&= r^2 \kab \ff, \\
\Phi_4&= \underline{P}(\Phi_3)
\end{split}
\eea
The quantities $\Phi_0, \Phi_1, \Phi_2$ contain information about the gravitational perturbation of the metric, i.e. about the Weyl curvature of the perturbed spacetime. The first quantity $\Phi_0$ is a rescaled version of $\a$, and being multiplied by $\kab^2$ it is of signature $0$. Then, $\Phi_1$ and $\Phi_2$ are respectively the first and the second derivative, through the operator $\underline{P}$, of $\Phi_0$, giving other two signature $0$ quantities. Observe that the last quantity $\Phi_2$ defined as
\beaa
\Phi_2&=&P(P(\Phi_0))=r^2\kab^{-2}e_3e_3\Phi_0+2r^2 \kab^{-1} (1+\omb\kab^{-1}) e_3\Phi_0 +\frac 12 r^2  \Phi_0=\\
&=&r^4( e_3e_3\a+(2\kab-6\omb )e_3\a+(\frac 1 2 \kab^2-8\kab\omb-4 e_3\omb +8\omb^2  )\a) =\qf
\eeaa
coincides with the $\qf$ obtained by Chandrasekhar transformation in \cite{stabilitySchwarzschild}.

The quantities $\Phi_3, \Phi_4$ contain information about the electromagnetic perturbation of the metric, i.e. about the Ricci curvature (or electromagnetic tensor) of the perturbed spacetime. The quantity $\Phi_3$ is a rescaled version of $\ff$, and being multiplied by $\kab$ it is of signature $0$. Then $\Phi_4$ is the first derivative, through the same operator $\underline{P}$, of $\Phi_3$, giving another signature $0$ quantity. The last quantity $\Phi_4$, in analogy to $\Phi_2$, is called $\qf^{\F}$.

\subsection{Wave equations for the curvature quantities $\Phi_0, \Phi_1, \Phi_2$}

We will derive the wave equations for the quantities $\Phi_0, \Phi_1, \Phi_2$.

\begin{proposition}\label{wavePhi0} Modulo $O(\ep^2)$,
\beaa
\square_\g \Phi_0&=& \frac 1 r \left(2\ka\kab-4\rho\right)  \Phi_1+\Big(- \frac 1 2  \ka\kab  -4\rho+4\rhoF^2+ 4 e_\th(\Phi)^2 \Big) \Phi_0+ \rhoF\left(\frac 2 rQ( \Phi_3)-4\rho \Phi_3\right)
\eeaa
\end{proposition}
\begin{proof} We have 
\beaa
\square_\g \Phi_0&=&\square_\g (r^2 \kab^2 \a)=\square_\g(r^2\kab^2) \a+r^2\kab^2 \square_\g\a-e_3(r^2\kab^2) e_4(\a)-e_4(r^2\kab^2) e_3(\a)
\eeaa
and using that 
\beaa
e_3(r^2 \kab^2)&=&e_3(r^2)\kab^2+r^2 e_3(\kab^2)=r^2 \kab^3+ r^2(- \kab^3-4\omb\kab^2)=-4\omb r^2\kab^2, \\
e_4(r^2 \kab^2)&=&e_4(r^2)\kab^2+r^2 e_4(\kab^2)=(r^2 \ka )\kab^2+ r^2(- \ka\kab^2+4\om\kab^2+4\rho\kab)=4\om r^2\kab^2+4\rho r^2\kab, \\
e_4(e_3(r^2 \kab^2))&=&e_4(-4\omb r^2\kab^2)=-4 e_4(\omb) r^2\kab^2-4\omb e_4(r^2\kab^2)=-4 e_4(\omb) r^2\kab^2-16\omb \om r^2\kab^2-16\omb \rho r^2\kab
\eeaa
we have 
\beaa
\square_\g (r^2 \kab^2)&=& -e_4(e_3( r^2\kab^2))+e_\th(e_\th( r^2\kab^2)) -\frac{1}{2}\kab e_4(r^2\kab^2) +\left(-\frac{1}{2}\ka+2\om\right) e_3(r^2\kab^2)+e_\th(\Phi)e_\th( r^2\kab^2) \\
&=& -(-4 e_4(\omb) r^2\kab^2-16\omb \om r^2\kab^2-16\omb \rho r^2\kab) -\frac{1}{2}\kab (4\om r^2\kab^2+4\rho r^2\kab) +\left(-\frac{1}{2}\ka+2\om\right) (-4\omb r^2\kab^2)\\
&=&\left( 2\omb\ka-2\om\kab+4 e_4(\omb) +8\omb \om\right) r^2\kab^2+(16\omb-2\kab )\rho r^2\kab
\eeaa

Therefore,
\beaa
\square_\g \Phi_0&=&\Bigg(\left( 2\omb\ka-2\om\kab+4 e_4(\omb) +8\omb \om\right) r^2\kab^2+(16\omb-2\kab )\rho r^2\kab \Bigg) \a\\
&&+r^2\kab^2 \Bigg(- 4 \omb e_4(\a)     +\left(2\ka+4\om\right) e_3(\a)+(- 4 e_4(\omb) +\frac{1}{2}\ka\kab - 10 \ka\omb +2\kab\om -8\om\omb - 4\rho + 4e_\th(\Phi)^2 )\a\Bigg)\\
&&-\left(-4\omb r^2\kab^2\right) e_4(\a)-\left(4\om r^2\kab^2+4\rho r^2\kab\right) e_3(\a)+2r^2\kab^2\rhoF\left(  e_4(\ff)+( \ka+2\om) \ff \right)
\eeaa
giving 
\beaa
\square_\g \Phi_0&=&r^2\Bigg(\left(2\ka\kab^2-4\rho\kab\right) e_3(\a)+\Big( \frac 1 2  \ka\kab^3  -6\rho\kab^2-8\omb\ka\kab^2+16\rho\kab\omb+4\kab^2\rhoF^2+ 4\kab^2 e_\th(\Phi)^2 \Big) \a+\\
&&+2\kab^2\rhoF\left(  e_4(\ff)+( \ka+2\om) \ff \right) \Bigg)
\eeaa
Using the following relations
\beaa
r^2\kab^2\a&=&\Phi_0, \qquad e_3\a=\frac {1}{r^2}\kab^{-2}e_3(\Phi_0)+4\omb \a, \qquad e_3\Phi_0=\frac 1 r \kab \Phi_1- \frac 1 2 \kab \Phi_0, \\
r^2\kab^2  e_4(\ff)+r^2( \ka\kab^2+2\om\kab^2) \ff&=&\frac 1 r \kab e_4(r^3\kab \ff)-2\rho\kab r^2\ff=\frac 1 r Q(\Phi_3)-2\rho \Phi_3
\eeaa
we obtain the desired identity.
\end{proof}

\begin{remark} Comparing Proposition \ref{squarea} to Proposition \ref{wavePhi0}, we can notice that the rescaled quantity $\Phi_0$ verifies a wave equation independent of the quantities $\om, \omb$, one of which can be made small in the ingoing or outgoing geodesic null frame. Therefore, the wave equation for $\Phi_0$ is more natural and frame-independent compared to the one for $\a$.
\end{remark}

\begin{proposition}\label{wavePhi1} We have modulo $O(\ep^2)$,
\beaa
\square_\g (\Phi_1)&=& \frac{1}{r}( \ka\kab-2\rho) \Phi_2+(- \ka\kab+6\rhoF^2+4e_\th\Phi^2)\Phi_1+(\frac 3 2 \rho+\rhoF^2 )r\Phi_0\\
&&+\rhoF \Bigg(\frac {2} {r} Q(\Phi_4)-Q(\Phi_3)-2\ka\kab \Phi_4+  r\left(6\rho+4\rhoF^2\right)\Phi_3\Bigg)
\eeaa
\end{proposition}
\begin{proof} We first compute $ e_3 (\square_\g \Phi_0)$, using Proposition \ref{wavePhi0}. 
\beaa
e_3 (\square_\g \Phi_0)&=&e_3\Bigg(\frac 1 r \left(2\ka\kab-4\rho\right)  \Phi_1+\Big(- \frac 1 2  \ka\kab  -4\rho+4\rhoF^2+ 4 e_\th(\Phi)^2 \Big) \Phi_0+ \rhoF\left(\frac 2 rQ( \Phi_3)-4\rho \Phi_3\right)\Bigg)=\\
&=&- \frac 1 2 r^{-1} \kab \left(2\ka\kab-4\rho\right)  \Phi_1+\frac 1 r \Big(2(-\frac 1 2 \ka\kab+2\omb\ka+2\rho)\kab+2\ka (-\frac 1 2 \kab^2-2\omb\kab)-4(-\frac 3 2 \kab \rho-\kab\rhoF^2)\Big)  \Phi_1+\\
&&+\frac 1 r \left(2\ka\kab-4\rho\right)  e_3\Phi_1+e_3\Big( -\frac 1 2 \ka\kab  -4\rho+4\rhoF^2+ 4e_\th(\Phi)^2\Big)\Phi_0+\Big( -\frac 1 2 \ka\kab  -4\rho+4\rhoF^2+ 4e_\th(\Phi)^2\Big)e_3\Phi_0+\\
&&+\rhoF \Bigg(-\kab \left(\frac 2 rQ( \Phi_3)-4\rho \Phi_3\right)- \frac 1 r \kab Q( \Phi_3)+\frac 2 r e_3(Q( \Phi_3))-4(-\frac 3 2 \kab \rho-\kab\rhoF^2) \Phi_3-4\rho e_3\Phi_3\Bigg)
\eeaa
and since 
\beaa
&&e_3\Big( -\frac 1 2 \ka\kab  -4\rho+4\rhoF^2+ 4e_\th(\Phi)^2\Big)= \frac 1 2 \ka\kab^2+5\rho\kab- 4\kab e_\th(\Phi)^2-4\kab\rhoF^2
\eeaa
applying Lemma \ref{e3underlineP} and writing $ e_3\Phi_0=\frac 1 r \kab \Phi_1 -\frac 12 \kab \Phi_0$, $e_3 \Phi_1=\frac 1 r \kab \Phi_2 -\frac 12 \kab \Phi_1$ and $e_3 \Phi_3=\frac 1 r \kab \Phi_4 -\frac 12 \kab \Phi_3$ we have 
\beaa
e_3 (\square_\g \Phi_0)&=&\frac 1 r \Big(-\frac 92\ka\kab^2+10\rho\kab+8\kab\rhoF^2+4\kab e_\th \Phi^2\Big)  \Phi_1+\frac {1}{ r^2} \left(2\ka\kab^2-4\rho\kab\right) \Phi_2+\\
&&+\left(\frac 34 \ka\kab^2+7\rho\kab- 6\kab e_\th(\Phi)^2-6\kab\rhoF^2\right)\Phi_0+\rhoF \Bigg(\frac {2}{ r^2}  \kab Q(\Phi_4)- \frac 4 r \kab Q( \Phi_3)-\frac 2 r \ka\kab^2\Phi_4+(12\rho\kab+4\kab\rhoF^2) \Phi_3\Bigg)
\eeaa
We have for $\Phi_1=\underline{P}(\Phi_0)$, applying Lemma \ref{squareP} to $f=\Phi_0$, and using Lemma \ref{wavePhi0} ,
\beaa
\square_\g (\Phi_1)&=& \frac{1}{r}( -\ka\kab+2\rho) \underline{P}(\underline{P}(\Phi_0))+(\frac 1 2 \ka\kab-4\rho-2\rhoF^2)\underline{P}(\Phi_0)+(\frac 1 2 \rho +\rhoF^2)r\Phi_0+ \frac 3 2r \square_\g(\Phi_0)+r\kab^{-1} e_3(\square_\g(\Phi_0))=\\
&=& \frac{1}{r}( -\ka\kab+2\rho) \Phi_2+(\frac 1 2 \ka\kab-4\rho-2\rhoF^2)\Phi_1+(\frac 1 2 \rho+\rhoF^2 )r\Phi_0+\\
&&+ \frac 3 2r \Bigg(\frac 1 r \left(2\ka\kab-4\rho\right)  \Phi_1+\Big(- \frac 1 2  \ka\kab  -4\rho+4\rhoF^2+ 4 e_\th(\Phi)^2 \Big) \Phi_0+ \rhoF\left(\frac 2 rQ( \Phi_3)-4\rho\Phi_3\right) \Bigg)+\\
&&+r\kab^{-1} \Bigg[ \frac 1 r \Big(-\frac 92\ka\kab^2+10\rho\kab+8\kab\rhoF^2+4\kab e_\th \Phi^2\Big)  \Phi_1+\frac {1}{ r^2} \left(2\ka\kab^2-4\rho\kab\right) \Phi_2+\\
&&+\left(\frac 34 \ka\kab^2+7\rho\kab- 6\kab e_\th(\Phi)^2-6\kab\rhoF^2\right)\Phi_0+\rhoF \Big(\frac {2}{ r^2}  \kab Q\Phi_4- \frac 4 r \kab Q( \Phi_3)-\frac 2 r \ka\kab^2\Phi_4+(12\rho\kab+4\kab\rhoF^2) \Phi_3\Big)\Bigg]=\\
&=& \frac{1}{r}( \ka\kab-2\rho) \Phi_2+(- \ka\kab+6\rhoF^2+4e_\th\Phi^2)\Phi_1+(\frac 3 2 \rho+\rhoF^2 )r\Phi_0\\
&&+\rhoF \Bigg(\frac {2} {r}Q(\Phi_4)-Q(\Phi_3)-2\ka\kab \Phi_4+  r\left(6\rho+4\rhoF^2\right)\Phi_3\Bigg)
\eeaa
as desired.
\end{proof}
 We derive the Regge-Wheeler type equation for the curvature term $\Phi_2=\qf$ with right hand side coupled to the electromagnetic components $\Phi_3$ and $\Phi_4$, multiplied by $\rhoF$.
 
 \begin{proposition} \label{regge1} We have modulo $O(\ep^2)$,
\beaa
\square_\g(\Phi_2)&=&\Big(- \ka\kab+6\rhoF^2+4 e_\th \Phi^2\Big) \Phi_2\\
&&+ \rhoF \Bigg(\frac {2} {r}  (Q(\underline{P}(\Phi_4)))-2Q(\Phi_4)+(-4\ka\kab+4\rho)\underline{P}(\Phi_4)+r (3\ka\kab+4\rhoF^2)\Phi_4+r^2(-6\rho-12\rhoF^2)\Phi_3\Bigg)
\eeaa
\end{proposition}
\begin{proof} We first compute $e_3(\square_\g \Phi_1)$, using Proposition \ref{wavePhi1}. We get
\beaa
e_3(\square_\g \Phi_1)&=&e_3 \Bigg( \frac{1}{r}( \ka\kab-2\rho) \Phi_2+(- \ka\kab+6\rhoF^2+4e_\th\Phi^2)\Phi_1+(\frac 3 2 \rho+\rhoF^2 )r\Phi_0\\
&&+\rhoF \Big(\frac {2} {r}Q(\Phi_4)-Q(\Phi_3)-2\ka\kab \Phi_4+  r\left(6\rho+4\rhoF^2\right)\Phi_3\Big) \Bigg)
\eeaa
which gives
\beaa
e_3(\square_\g \Phi_1)&=& -\frac 1 2 r^{-1}\kab \Big( \ka\kab-2\rho\Big) \Phi_2+ \frac{1}{r}\Big( (-\frac 1 2 \ka\kab+2\omb\ka+2\rho)\kab+\ka (-\frac 1 2 \kab^2-2\omb\kab)-2(-\frac 3 2 \kab \rho-\kab \rhoF^2)\Big) \Phi_2+\\
&&+ \frac{1}{r}\Big( \ka\kab-2\rho\Big) e_3\Phi_2+\Big(- (-\frac 1 2 \ka\kab+2\omb\ka+2\rho)\kab- \ka (-\frac 1 2 \kab^2-2\omb\kab)-12\kab \rhoF^2-4 \kab e_\th\Phi^2\Big)\Phi_1+\\
&&+\Big(- \ka\kab+6\rhoF^2+4 e_\th\Phi^2\Big)e_3\Phi_1+r(\frac 3 4 \kab \rho+\frac 1 2 \kab \rhoF^2) \Phi_0+r(\frac 3 2  (-\frac 3 2 \kab\rho-\kab\rhoF^2) -2\kab\rhoF^2)\Phi_0+\\
&&+r(\frac 3 2 \rho+\rhoF^2) e_3(\Phi_0)+\rhoF \Bigg[-\frac {2} {r}\kab Q\Phi_4+\kab  Q( \Phi_3)+ 2\ka\kab^2\Phi_4+r(-6\kab\rho-4\kab\rhoF^2)\Phi_3+\\
&&+ e_3(\frac {2} {r}) Q\Phi_4+\frac {2} {r} e_3(Q(\Phi_4))-  e_3(Q( \Phi_3))-2 e_3\left(\ka\kab\right)\Phi_4-2\ka\kab \ e_3\Phi_4+  \frac 1 2 \kab r\left(6\rho+4\rhoF^2\right)\Phi_3+\\
&&+  re_3\left(6\rho+4\rhoF^2\right)\Phi_3+  r\left(6\rho+4\rhoF^2\right)e_3\Phi_3\Bigg]=\\
&=& \frac{1}{r}\Big( -\frac 3 2 \ka\kab^2+6\rho\kab+2\kab \rhoF^2\Big) \Phi_2+ \frac{1}{r}\Big( \ka\kab-2\rho\Big) e_3\Phi_2+\\
&&+\Big(\ka\kab^2- 2\rho\kab-12\kab\rhoF^2-4 \kab e_\th\Phi^2\Big)\Phi_1+\Big(- \ka\kab+6\rhoF^2+4 e_\th\Phi^2\Big)e_3\Phi_1+r(-\frac 3 2 \kab \rho-3\kab\rhoF^2 )\Phi_0+\\
&&+r(\frac 3 2 \rho+\rhoF^2) e_3(\Phi_0)+\rhoF \Bigg[-\frac {3} {r}\kab Q(\Phi_4)+\kab Q( \Phi_3)+ (4\ka\kab^2-4\rho\kab)\Phi_4+r(-12\kab\rho-16\kab\rhoF^2)\Phi_3+\\
&&+\frac {2} {r} e_3Q(\Phi_4)-  e_3Q( \Phi_3)-2\ka\kab \ e_3\Phi_4+  r\left(6\rho+4\rhoF^2\right)e_3\Phi_3\Bigg]
\eeaa
Using Lemma \ref{e3underlineP} and writing $ e_3\Phi_0=\frac 1 r \kab \Phi_1 -\frac 12 \kab \Phi_0$, $e_3 \Phi_1=\frac 1 r \kab \Phi_2 -\frac 12 \kab \Phi_1$,  $e_3 \Phi_2=\frac 1 r \kab \underline{P}(\Phi_2) -\frac 12 \kab \Phi_2$, and $e_3 \Phi_3=\frac 1 r \kab \Phi_4 -\frac 12 \kab \Phi_3$, $e_3 \Phi_4=\frac 1 r \kab \underline{P}(\Phi_4) -\frac 12 \kab \Phi_4$ we have 
\beaa
e_3(\square_\g \Phi_1)&=&\frac{1}{r^2}\Big( \ka\kab^2-2\rho\kab\Big) \underline{P}(\Phi_2)+\frac{1}{r}\Big( -3 \ka\kab^2+7\rho\kab+8\kab\rhoF^2+4\kab e_\th \Phi^2\Big) \Phi_2+\Big(\frac 3 2 \ka\kab^2- \frac 1 2 \rho\kab-14\kab\rhoF^2-6 \kab e_\th\Phi^2\Big)\Phi_1+\\
&&+r(-\frac 9 4 \kab \rho-\frac 7 2 \kab \rhoF^2) \Phi_0+\rhoF \Bigg[\frac {2} {r^2} Q\underline{P}\Phi_4-\frac {5} {r}\kab Q\Phi_4+\frac {1} {r} (-4\ka\kab^2+4\rho\kab)\underline{P}\Phi_4+\frac 3 2 \kab  Q( \Phi_3)+\\
&&+ (6\ka\kab^2+4\kab\rhoF^2)\Phi_4+r(-15\kab\rho-18\kab\rhoF^2)\Phi_3\Bigg]
\eeaa
 Applying Lemma \ref{squareP} to $\Phi_2=P(\Phi_1)$, 
\beaa
\square_\g(\Phi_2)&=&\frac{1}{r}( -\ka\kab+2\rho) \underline{P}(\underline{P}(\Phi_1))+(\frac 1 2 \ka\kab-4\rho-2\rhoF^2)\underline{P}(\Phi_1)+(\frac 1 2 \rho +\rhoF^2)r\Phi_1+ \frac 3 2r \square_\g(\Phi_1)+r\kab^{-1} e_3(\square_\g(\Phi_1))=\\
&=&\frac{1}{r}( -\ka\kab+2\rho) \underline{P}(\Phi_2)+(\frac 1 2 \ka\kab-4\rho-2\rhoF^2)\Phi_2+(\frac 1 2 \rho+\rhoF^2) r \Phi_1+\\
&&+ \frac 3 2r \Bigg[\frac{1}{r}( \ka\kab-2\rho) \Phi_2+(- \ka\kab+6\rhoF^2+4e_\th\Phi^2)\Phi_1+(\frac 3 2 \rho+\rhoF^2 )r\Phi_0\\
&&+\rhoF \Big(\frac {2} {r}Q(\Phi_4)-Q(\Phi_3)-2\ka\kab \Phi_4+  r\left(6\rho+4\rhoF^2\right)\Phi_3\Big)\Bigg]+\\
&&+r\kab^{-1} \Bigg[\frac{1}{r^2}\Big( \ka\kab^2-2\rho\kab\Big) P(\Phi_2)+\frac{1}{r}\Big( -3 \ka\kab^2+7\rho\kab+8\kab\rhoF^2+4\kab e_\th \Phi^2\Big) \Phi_2+\Big(\frac 3 2 \ka\kab^2- \frac 1 2 \rho\kab-14\kab\rhoF^2-6 \kab e_\th\Phi^2\Big)\Phi_1+\\
&&+r(-\frac 9 4 \kab \rho-\frac 7 2 \kab \rhoF^2) \Phi_0+\rhoF \Big[\frac {2} {r^2} \kab QP\Phi_4-\frac {5} {r}\kab Q\Phi_4+\frac {1} {r} (-4\ka\kab^2+4\rho\kab)\underline{P}\Phi_4+\frac 3 2 \kab  Q( \Phi_3)+\\
&&+ (6\ka\kab^2+4\kab\rhoF^2)\Phi_4+r(-15\kab\rho-18\kab\rhoF^2)\Phi_3\Big]\Bigg]
\eeaa
which gives
\beaa
\square_\g(\Phi_2)&=&\Big(- \ka\kab+6\rhoF^2+4 e_\th \Phi^2\Big) \Phi_2\\
&&+ \rhoF \Bigg(\frac {2} {r}  (Q(\underline{P}(\Phi_4)))-2Q(\Phi_4)+(-4\ka\kab+4\rho)\underline{P}(\Phi_4)+r (3\ka\kab+4\rhoF^2)\Phi_4+r^2(-6\rho-12\rhoF^2)\Phi_3\Bigg)
\eeaa
as desired.
\end{proof}

\begin{remark} Notice that the wave equation given by Proposition \ref{regge1}, for $\Phi_2=\qf$ and $\Phi_4=\qf^{\F}$ has the form
\beaa
\square_\g \qf= V_1 \qf+ e \mathcal{M}(\qf^{\F}, \partial \qf^{\F}, \partial \partial \qf^{\F})+e (l.o.t.(\qf^\F))+e^2 (l.o.t.(\qf))
\eeaa
of the first equation of \eqref{schemesystem}. Indeed, $\Phi_3$ is a lower order term with respect to $\qf^\F$.
\end{remark}

\subsection{Wave equations for the electromagnetic quantities $\Phi_3, \Phi_4$}

We compute the wave equations for the electromagnetic terms $\Phi_3, \Phi_4$.

\begin{proposition}\label{wavePhi3} We have modulo $O(\ep^2)$,
\beaa
\square_\g \Phi_3&=&\frac 1 r\Big( \ka\kab-2\rho \Big)\Phi_4+\Big(-\ka\kab-3\rho +4 e_\th\Phi^2 \Big)\Phi_3+  \rhoF \Big (-\frac 2 r  \Phi_1- \Phi_0 \Big)
\eeaa
\end{proposition}
\begin{proof} We have for $\Phi_3=r^2 \kab \ff$,
\beaa
\square_\g \Phi_3&=& \square_\g (r^2) \kab \ff+r^2 (\square(\kab)\ff+\kab \square( \ff)-e_3(\kab)e_4( \ff)-e_4(\kab) e_3(\ff))-e_3(r^2) (e_4(\kab) \ff+\kab e_4(\ff))-e_4(r^2) (e_3(\kab) \ff+\kab e_3(\ff))
\eeaa
As in Proposition \ref{wavePhi0}, using Proposition \ref{squareff}, we have 
\beaa
\square_\g \Phi_3&=&\Big(-5\rho-2\omb\ka+4\rho\omb\kab^{-1} +4 e_\th\Phi^2 \Big)\Phi_3+r^2 \Big( \ka\kab-2\rho \Big)e_3(\ff)+\\
&&+ 2r^2  \rhoF \Big (-\kab e_3 (\a)+ (-\kab^{2}+4\omb\kab)\a \Big) 
\eeaa
Writing $e_3(\ff)=\frac {1}{r^2}\kab^{-2}e_3(\Phi_3)+4\omb \ff$ 
and $e_3\a=\frac {1}{r^2}\kab^{-2}e_3(\Phi_0)+4\omb \a$, and $e_3\Phi_0=\frac 1 r \kab \Phi_1 -\frac 12 \kab \Phi_0$ we have
\beaa
\square_\g \Phi_3&=&\frac 1 r \Big( \ka\kab-2\rho \Big) P \Phi_3+\Big(-\ka\kab-3\rho +4 e_\th\Phi^2 \Big)\Phi_3+  \rhoF \Big (-\frac 2 r  \Phi_1- \Phi_0 \Big) 
\eeaa
as desired.
\end{proof}
We derive the Regge-Wheleer type equation for the quantity $\Phi_4=\qf^\F$, with on the right hand side the curvature multiplied by $\rhoF$.

\begin{proposition}\label{regge2} We have modulo $O(\ep^2)$,
\beaa
\square_\g (\Phi_4)&=&\left(- \ka\kab-3\rho+4e_\th\Phi^2\right)\Phi_4+\rhoF  \Big(-\frac {2}{ r} \Phi_2+ \rhoF \left(4r\Phi_3 \right)\Big)
\eeaa
\end{proposition}
\begin{proof} We first compute $e_3 (\square_\g \Phi_3)$, using Proposition \ref{wavePhi3},
\beaa
e_3 (\square_\g \Phi_3)&=&e_3\Bigg(\frac 1 r\Big( \ka\kab-2\rho \Big)\Phi_4+\Big(-\ka\kab-3\rho +4 e_\th\Phi^2 \Big)\Phi_3+  \rhoF \Big (-\frac 2 r  \Phi_1- \Phi_0 \Big)\Bigg)=\\
&=&- \frac 1 2 r^{-1} \kab \left(\ka\kab-2\rho\right)  \Phi_4+\frac 1 r \left((-\frac 1 2 \ka\kab+2\omb\ka+2\rho)\kab+\ka (-\frac 1 2 \kab^2-2\omb\kab)-2(-\frac 3 2 \kab \rho-\kab\rhoF^2)\right)  \Phi_4+\frac 1 r \left(\ka\kab-2\rho\right)  e_3\Phi_4+\\
&&+e_3\Big( - \ka\kab  -3\rho+ 4e_\th(\Phi)^2\Big)\Phi_3+\Big( - \ka\kab  -3\rho+ 4e_\th(\Phi)^2\Big)e_3\Phi_3+\rhoF \Bigg(\frac 2 r \kab \Phi_1+\kab \Phi_0 -e_3(\frac 2 r)  \Phi_1-\frac 2 r  e_3\Phi_1- e_3\Phi_0\Bigg)
\eeaa
Since 
\beaa
&&e_3\Big( - \ka\kab  -3\rho+ 4e_\th(\Phi)^2\Big)= - e_3\ka\kab- \ka e_3\kab  -3e_3\rho+ 8e_\th(\Phi) e_3 e_\th \Phi=\\
&=& - (-\frac 1 2 \ka\kab+2\omb\ka+2\rho)\kab- \ka (-\frac 1 2\kab^2-2\omb\kab)  -3(-\frac 3 2 \kab\rho-\kab\rhoF^2)- 4\kab e_\th(\Phi)^2=\\
&=&\ka\kab^2+\frac 5 2 \rho\kab+3\kab\rhoF^2- 4\kab e_\th(\Phi)^2
\eeaa
and writing $ e_3\Phi_0=\frac 1 r \kab \Phi_1 -\frac 12 \kab \Phi_0$, $e_3 \Phi_1=\frac 1 r \kab \Phi_2 -\frac 12 \kab \Phi_1$, $e_3\Phi_3=\frac 1 r \kab \Phi_4 -\frac 12 \kab \Phi_3$, $e_3\Phi_4=\frac 1 r \kab \underline{P}\Phi_4 -\frac 12 \kab \Phi_4$  we have 
\beaa
e_3 (\square_\g \Phi_3)&=&\frac 1 r \left(-3 \ka\kab^2+4\rho\kab+2\kab\rhoF^2+4\kab e_\th\Phi^2\right)  \Phi_4+\frac {1}{ r^2} \left(\ka\kab^2-2\rho\kab \right)  \underline{P}\Phi_4+\left(\frac 3 2 \ka\kab^2+4\rho\kab+3\kab\rhoF^2- 6\kab e_\th(\Phi)^2\right)\Phi_3+\\
&&+\rhoF \Bigg(-\frac {2}{ r^2}  \kab \Phi_2+\frac 3 r \kab \Phi_1+\frac 3 2 \kab \Phi_0 \Bigg)
\eeaa
Applying Lemma \ref{squareP}, we have 
\beaa
\square_\g (\Phi_4)&=&\frac{1}{r}( -\ka\kab+2\rho) \underline{P}(\underline{P}(\Phi_3))+(\frac 1 2 \ka\kab-4\rho-2\rhoF^2)\underline{P}(\Phi_3)+\left(\frac 1 2 \rho+\rhoF^2 \right)r\Phi_3+ \frac 3 2r \square_\g(\Phi_3)+\kab^{-1}r e_3(\square_\g(\Phi_3))=\\
&=&\frac{1}{r}( -\ka\kab+2\rho) \underline{P}(\Phi_4)+(\frac 1 2 \ka\kab-4\rho-2\rhoF^2)\Phi_4+\left(\frac 1 2 \rho+\rhoF^2 \right)r\Phi_3+\\
&&+ \frac 3 2r \Big(\frac 1 r\Big( \ka\kab-2\rho \Big)\Phi_4+\Big(-\ka\kab-3\rho +4 e_\th\Phi^2 \Big)\Phi_3+  \rhoF \Big (-\frac 2 r  \Phi_1- \Phi_0 \Big) \Big)+\\
&&+\kab^{-1}r \Big(\frac 1 r \left(-3 \ka\kab^2+4\rho\kab+2\kab\rhoF^2+4\kab e_\th\Phi^2\right)  \Phi_4+\frac {1}{ r^2} \left(\ka\kab^2-2\rho\kab \right)  \underline{P}\Phi_4 +\\
&&+\left(\frac 3 2 \ka\kab^2+4\rho\kab+3\kab\rhoF^2- 6\kab e_\th(\Phi)^2\right)\Phi_3+\rhoF \Bigg(-\frac {2}{ r^2}  \kab \Phi_2+\frac 3 r \kab \Phi_1+\frac 3 2 \kab \Phi_0 \Bigg) \Big)
\eeaa
which gives 
\beaa
\square_\g (\Phi_4)&=&(- \ka\kab-3\rho+4e_\th\Phi^2)\Phi_4+\rhoF  \Big(-\frac {2}{ r} \Phi_2+ 4r\rhoF \Phi_3 \Big)
\eeaa
as desired.
\end{proof}

\begin{remark} Notice that the wave equation given by Proposition \ref{regge2}, for $\Phi_2=\qf$ and $\Phi_4=\qf^{\F}$ has the form
\beaa
\square_\g \qf^{\F}=V_2 \qf^{\F}+e \mathcal{M}(\qf) +e^2( l.o.t.(\qf^\F))
\eeaa
of the second equation of \eqref{schemesystem}.
\end{remark}

Using the wave equation for $\Phi_4$, we can simplify the wave equation for $\Phi_2$ in Proposition \ref{regge1}, since the derivative $\underline{P} P$ is related to $\square_2:=\square_\g -(2)^2 e_\th\Phi^2$ in the following way:
\beaa
\frac{1}{r^2}Q(\underline{P}(\Phi_4))&=&-\square_2 \Phi_4+\frac{1}{r}(  \ka\kab-2\rho) \underline{P}(\Phi_4)+\lapp_2\Phi_4 +\rho \Phi_4, 
\eeaa
where $\lapp_2$ is the Laplacian on the spheres $S$ of the foliation of the spacetime.
Using Proposition \ref{regge2}, we can write
\beaa
Q(\underline{P}(\Phi_4))&=&r(  \ka\kab-2\rho) \underline{P}(\Phi_4)+r^2\lapp_2\Phi_4+r^2(\ka\kab+4\rho) \Phi_4+ 2 r \rhoF \Phi_2-4r^3\rhoF^2\Phi_3 
\eeaa
giving 
\beaa
\square_\g \Phi_2&=&\Big(- \ka\kab+10\rhoF^2+4 e_\th \Phi^2\Big) \Phi_2\\
&&+ \rhoF \Bigg(2r\lapp_2\Phi_4-2Q(\Phi_4)-2\ka\kab \underline{P}(\Phi_4)+r (5\ka\kab+8\rho+4\rhoF^2)\Phi_4+r^2(-6\rho-20\rhoF^2)\Phi_3\Bigg)+\\
&&+\rhoF^2\left(-4r\Phi_1-2r^2\Phi_0 \right)
\eeaa

\subsection{The system of coupled wave equations}

Writing the five equations together, using \eqref{rhoFcharge} to write $\rhoF=\frac{e}{r^2} +O(\ep)$ in the coupling terms, we found the following system of equations modulo $O(\ep^2)$:

\beaa\label{system}
\left(\square_2+ \frac 1 2  \ka\kab  +4\rho-4\rhoF^2 \right)\Phi_0&=& \frac 1 r \left(2\ka\kab-4\rho\right)  \Phi_1+ e\left(\frac {2}{ r^3}Q( \Phi_3)-\frac{4}{r^2}\rho\Phi_3\right),\\
\Big(\square_2+ \ka\kab-6\rhoF^2 \Big)\Phi_1&=&\frac{1}{r}( \ka\kab-2\rho) \Phi_2+r\left(\frac 3 2 \rho +\rhoF^2\right)\Phi_0+\\
&&+e \Bigg(\frac {2} {r^3} Q(\Phi_4)-  \frac {1}{r^2}Q( \Phi_3)- \frac {2}{r^2}\ka\kab\Phi_4+\frac 1 r (6\rho+4\rhoF^2) \Phi_3\Bigg), \\
\Big(\square_2+ \ka\kab-10\rhoF^2\Big)\Phi_2&=& e\Bigg(\frac{2}{r}\lapp_2\Phi_4-\frac{2}{r^2}Q(\Phi_4)-\frac{2}{r^2}\ka\kab \underline{P}(\Phi_4) + \frac 1 r \left(5\ka\kab+8\rho+4\rhoF^2\right)\Phi_4-6\rho \Phi_3\\
&&+e\left(-\frac{4}{r^3}\Phi_1-\frac{2}{r^2} \Phi_0 \right)+e^2 \left( -\frac{20}{r^4} \Phi_3\right)\Bigg), \\
\Big(\square_2+\ka\kab+3\rho \Big)\Phi_3&=&\frac 1 r\Big( \ka\kab-2\rho \Big)\Phi_4+  e \Big (-\frac {2}{ r^3}  \Phi_1- \frac{1}{r^2}\Phi_0 \Big), \\
\Big(\square_2+\ka\kab+3\rho \Big)\Phi_4&=&e \Big(-\frac {2}{ r^3} \Phi_2+\frac{4e}{r^3}\Phi_3 \Big)
\eeaa
where $\square_2=\square_\g -(2)^2 e_\th\Phi^2$ is the wave operator applied to $2$-reduced scalars.

\begin{remark}\label{spin-2} A complete analogous system holds for the spin $-2$ quantities $\aa$. Defining the operator $Pf=r \ka^{-1} e_4f +\frac 12 rf $ and the $O(\ep^2)$ invariant quantity $\underline{\ff}=\dds_2\bbF+\vthb\rhoF$, then the quantities $\widetilde{\Phi_0}=r^2 \ka^2 \aa$, $\widetilde{\Phi_1}=P(\widetilde{\Phi_0})$, $\widetilde{\Phi_2}=P(\widetilde{\Phi_1})$, $\widetilde{\Phi_3}=r^2 \ka \underline{\ff}$, $\widetilde{\Phi_4}=P(\widetilde{\Phi_3})$ verify the same system above, with $\underline{Q}f=r\ka  e_3f+\frac 1 2r \ka\kab f$.
\end{remark}

Selecting the third and fifth equation we have the system of Regge-Wheeler type equations for $\qf=\Phi_2$ and $\qf^{\F}=\Phi_4$ modulo $O(\ep^2)$, as announced in \eqref{schemesystem}. We summarize it in the following theorem.

\begin{theorem} Let $(\M, \g, \Z)$ be an axially symmetric polarized spacetime solution of the Einstein-Maxwell equation \eqref{Einsteineq}, which is a $O(\ep)$-perturbation of Reissner-Nordstr{\"o}m spacetime. Then there exist $O(\ep^2)$-invariant quantities $\qf$ and $\qf^{\F}$ related to the Weyl curvature and to the Ricci curvature respectively that verify the following coupled system of wave equations, modulo $O(\ep^)$ terms,
\bea\label{finalsystem}
\begin{cases}
\Big(\square_2+ \ka\kab-10\rhoF^2\Big)\qf= e\Bigg(\frac{2}{r}\lapp_2\qf^{\F}-\frac{2}{r^2}Q\qf^{\F}-\frac{2}{r^2}\ka\kab \underline{P}\qf^{\F} + \frac 1 r \left(5\ka\kab+8\rho+4\rhoF^2\right)\qf^{\F}\Bigg)+e(l.o.t.)_1, \\
\Big(\square_2+\ka\kab+3\rho\Big)\qf^{\F}=e \Bigg(-\frac {2}{ r^3} \qf\Bigg) +e^2 (l.o.t.)_2
\end{cases}
\eea
where $\underline{P}$ and $Q$ are rescaled null derivatives, as defined in \eqref{operators}, and $(l.o.t.)_1$ and $(l.o.t.)_2$ are lower order terms with respect to $\qf$ and $\qf^{\F}$, explicitely,
\beaa
(l.o.t.)_1&=&-6\rho \Phi_3+e\left(-\frac{4}{r^3}\Phi_1-\frac{2}{r^2} \Phi_0 \right)+e^2 \left( -\frac{20}{r^4} \Phi_3\right), \\
(l.o.t.)_2&=&\frac{4}{r^3} \Phi_3 
\eeaa
\end{theorem}

\begin{remark} The structure of the coupling in \eqref{finalsystem} doesn't depend on the polarization of the metric, as observed in \cite{Steffen}. See Appendix. 
\end{remark}

\subsection{Case of perturbation of Schwarzschild spacetime}\label{Sch}

Coupled gravitational and electromagnetic perturbations of Reissner-Nordstr{\"o}m spacetime are clearly a generalization of gravitational perturbations of Schwarzschild spacetime as solution to the vacuum Einstein equation \eqref{vacuum}, as treated in \cite{stabilitySchwarzschild}. In this case, the electromagnetic quantities and the quasi-local charge in \eqref{system} vanish identically, and the system reduces to the first three equations:
\beaa
\left(\square_2+ \frac 1 2  \ka\kab  +4\rho \right)\Phi_0&=& \frac 1 r \left(2\ka\kab-4\rho\right)  \Phi_1+ O(\ep^2), \\
\Big(\square_2+ \ka\kab \Big)\Phi_1&=&\frac{1}{r}( \ka\kab-2\rho) \Phi_2+\frac 3 2 \rho r\Phi_0+O(\ep^2), \\
\Big(\square_2+ \ka\kab\Big)\Phi_2&=& O(\ep^2)
\eeaa
which are the linear parts of the equations obtained in Appendix A.3.2 of \cite{stabilitySchwarzschild}. In particular, the last equation, for $\qf=\Phi_2$ is
\beaa
\square_2 \qf + \ka\kab \qf =\err[\square_\g \qf]
\eeaa 
which is the main equation used in \cite{stabilitySchwarzschild}, to derive decay estimates for $\qf$, and subsequently for $\a$ and all other curvature and connection coefficients quantities. 

In the case of coupled gravitational and electromagnetic perturbation of Schwarzschild spacetime, namely perturbation of Schwarzschild as solution to the Einstein-Maxwell equation \eqref{Einsteineq}, the system \eqref{finalsystem} simplifies. Perturbing the background Schwarzschild, being a vacuum spacetime, we have $\rhoF=O(\ep)$. Therefore, the main equation for the curvature $\qf$ is unchanged, but the right hand side is given by quadratic terms only, i.e.
\beaa
\square_2 \qf + \ka\kab \qf =O(\ep^2)
\eeaa 
Again since $\rhoF=O(\ep)$, using Proposition \ref{prop:transformations1}, in the case of perturbation of Schwarzschild the extreme components of the electromagnetic tensor $\bF, \bbF$ are $O(\ep^2)$-invariant quantities, and verify the Teukolsky equation
\beaa
\square_\g \bF&=&-2\omb e_4\bF +\left(\ka+2\om\right) e_3(\bF)+(\frac 1 4 \ka \kab-3\ka\omb+\om\kab+e_\th\Phi^2-2e_4\omb) \bF +O(\ep^2)
\eeaa
as derived in Proposition \ref{squareff}. As in \cite{Federico}, we can define a Chandrasekhar-type transformation at the level of one derivative along the ingoing null direction to obtain a Regge-Wheeler equation. Defining $\frak{l}=r^2(e_3 \aF_\th + \frac 1 2  \kab \aF_\th)$, in the case of $\rhoF=O(\ep)$ in a frame for which $\om=O(\ep)$, then $\frak{l}$ verifies the equation 
\beaa
\square_{1} \frak{l}&=-\frac 1 4 \ka \kab \frak{l} +O(\ep^2)
\eeaa
The system is therefore given by equations which are decoupled at the linear level
\beaa
\begin{cases}
\square_2 \qf + \ka\kab \qf =O(\ep^2), \\
\square_{1} \frak{l}+\frac 1 4 \ka \kab \frak{l} =O(\ep^2)
\end{cases}
\eeaa

\begin{remark} It is only in the case of gravitational and electromagnetic perturbations of Reissner-Nordstr{\"o}m spacetime that we find a non-trivial coupling for the linear terms of the equations for $\qf$ and $\qf^\F$ as described in system \eqref{finalsystem}. If the coupled gravitational and electromagnetic perturbations of Kerr-Newman spacetime would have a structure similar to the one here presented is an open question to be addressed.
\end{remark}

\appendix
\section{System of equations without polarization}
In this appendix, we will not assume polarization of the metric or axial symmetry. This appendix is based on computations done through computer algebra by Steffen Aksteiner.

Consider a null pair $e_3, e_4$  on $(\M, \g)$   and, at every point $p\in \M$
 the horizontal  space $S=\{ e_3, e_4\}^\perp$. Let $\ga $ the metric induced on   $S$.  By definition,  for all
     $X, Y\in T_S  \M$, i.e. vectors in $\M$ tangent to $S$, $\ga(X, Y)= \g(X, Y)$. 
 For any $Y\in T(\M)$ we define its horizontal projection by
 \bea
  Y^\perp = Y+\frac 12 \g( Y, e_3) e_4+\frac 12 \g(Y, e_3) e_4
  \eea
  \begin{definition}
  \label{definition:Shorizonthal}
A  $k$-covariant tensor-field $U$ is said to be  $S$-horizontal,  $U\in \T^k_S(\M)$,
if  for any $X_1,\ldots X_k$ we have,
\beaa
U(Y_1,\ldots Y_k)=U(Y^\perp _1,\ldots  Y^\perp_k)
\eeaa
\end{definition}
 \begin{definition}
Given  $X\in \T(\M)$ and $Y \in \T_S(\M)$ we define,
\beaa
\Db_X Y&:=& ( \D_X Y)^\perp
\eeaa

\end{definition} 
\begin{remark}
In the particular case when   $S$ is integrable and   both $X, Y\in \T_S\M$  then
$\Db_X Y$ is the standard  induced covariant differentiation on $S$. \end{remark}

 \begin{definition}
 \label{definition:S-covariantderivative}
 Given a general, covariant,  $S$- horizontal tensor-field  $U$
 we define its horizontal covariant derivative according to
 the formula,
 \bea
 \Db_X U(Y_1,\ldots Y_k)=X (U(Y_1,\ldots Y_k))&-&U(\Db_X Y_1,\ldots Y_k)-\ldots- U(Y_1,\ldots \Db_XY_k).
 \eea
 where $X\in \T\M$ and $Y_1,\ldots Y_k\in \T_S\M$.
 \end{definition}
 
 \begin{definition} Given $\Psi$ a 2 S-horizontal tensor, we define the wave operator $\squared_\g$ applied to $\Psi$ by 
 \beaa
 \squared_\g\Psi_{AB}:= \g^{\mu\nu} \Db_\mu\Db_ \nu \Psi_{AB}
 \eeaa
 \end{definition}

Recall the definition of spacetime Ricci coefficients and spacetime null curvature components in \eqref{def1}, \eqref{def2}, \eqref{def3}. In particular recall 
\beaa
\aS_{AB}&=& W_{A4B4},
\eeaa
The tensorial version\footnote{Indeed, $\,^{(1+3)} \hspace{-2.2pt}\ff_{\th\th}=-\,^{(1+3)} \hspace{-2.2pt}\ff_{\vphi\vphi}=\ff$, $\,^{(1+3)} \hspace{-2.2pt}\ff_{\th\vphi}=0$ in the case of polarized metric.}  of the invariant quantity $\ff$ introduced in Section \ref{defofff} is given by 
\beaa
\,^{(1+3)} \hspace{-2.2pt}\ff_{AB}=2\DDs_2 \bF_{AB}+2\rhoF \chih_{AB}
\eeaa
where $\DDs_2 \bF_{AB}=- \nabb_{(A} \bF_{B)}+\g_{AB}\ \div \bF $, as introduced in \cite{stabilitySchwarzschild}.

We define the tensorial versions of the operators $\underline{P}$ and $Q$ introduced in Section \ref{acknSteff} as 
\bea\label{tensorialPs}
\underline{P}(\Psi_{AB})=r\kab^{-1} e_3 \Psi_{AB}+\frac 1 2 r \Psi_{AB}, \\
Q(\Psi_{AB})=r\kab e_4 \Psi_{AB}+\frac 1 2 r \ka\kab \Psi_{AB}
\eea
where $\ka=\trchb$ and $\kab=\trchbS$.

We finally define the main quantities that verify the system of wave equations.
\begin{definition}
The tensorial quantities $\qf_{AB}$ and $\qf^{\F}_{AB}$ are defined by
\beaa
\,^{(1+3)} \hspace{-2.2pt}\qf_{AB}&=& \underline{P}\left(\underline{P}\left(r^2 \kab^2\aS_{AB}\right)\right), \qquad \,^{(1+3)} \hspace{-2.2pt}\qf^{\F}_{AB}=\underline{P}\left( r^2 \kab\,^{(1+3)} \hspace{-2.2pt}\ff_{AB}\right)
\eeaa
where $\kab =\trchbS$.
\end{definition}
We summarize in the following theorem the system of tensorial wave equation that is verified by the two quantities $\qf_{AB}$ and $\qf^\F_{AB}$.

\begin{theorem} Let $(\M, \g)$ be a spacetime solution of the Einstein-Maxwell equation \eqref{Einsteineq}, which is a $O(\ep)$-perturbation of Reissner-Nordstr{\"o}m spacetime. Then the tensorial quantities $\qf_{AB}$ and $\qf^\F_{AB}$ verify the following coupled system of wave equations, for $A,B=1,2$, modulo $O(\ep^2)$,
\bea\label{finalsystemtensorial}
\begin{cases}
\Big(\squared_\g+ \ka\kab-10\rhoF^2\Big)\qf_{AB}= e\Bigg(\frac{2}{r}\lapp_2\qf^{\F}_{AB}-\frac{2}{r^2}Q(\qf^{\F}_{AB})-\frac{2}{r^2}\ka\kab\underline{P}(\qf^{\F}_{AB}) + \frac 1 r \left(5\ka\kab+8\rho+4\rhoF^2\right)\qf^{\F}_{AB}\Bigg)+e(l.o.t.)_1, \\
\Big(\squared_\g+\ka\kab+3\rhoS\Big)\qf^{\F}_{AB}=e \Bigg(-\frac {2}{ r^3} \qf_{AB}\Bigg) +e^2 (l.o.t.)_2
\end{cases}
\eea
where $P$ and $\underline{P}$ are tensorial null derivatives, as defined in \eqref{tensorialPs}, $(\lapp_2\Psi)_{AB}=\ga^{CD}\nabb_{C}\nabb_{D}\Psi_{AB}$ and $(l.o.t.)_1$ and $(l.o.t.)_2$ are lower order terms with respect to $\qf_{AB}$ and $\qf^{\F}_{AB}$.
\end{theorem}

The system has the same structure as in the case of polarized metrics. Since no symmetries are assumed in this case, the analysis of the system \eqref{finalsystemtensorial} can be used to derive linear stability of Reissner-Nordstr{\"o}m spacetime, as a generalization to \cite{DHR}.

\begin{flushleft}
\small{DEPARTMENT OF MATHEMATICS, COLUMBIA UNIVERSITY} \\
\textit{E-mail address}: \href{mailto:egiorgi@math.columbia.edu}{egiorgi@math.columbia.edu}
\end{flushleft}

\end{document}